\documentclass[10pt, journal]{IEEEtran}
\usepackage{cite}
\usepackage{amsmath,amssymb,amsfonts}
\usepackage{graphicx}
\usepackage{algorithm}
\usepackage[hidelinks]{hyperref}
\usepackage{textcomp}

\usepackage[utf8]{inputenc} 
\usepackage[T1]{fontenc}    
\usepackage{amsmath,amssymb,amsfonts}
\usepackage{algpseudocode}
\usepackage{todonotes}
\usepackage{multirow}
\usepackage{url}
\usepackage{cite}
\usepackage{textcomp}
\usepackage{xcolor}
\usepackage{optidef}
\usepackage{arydshln}
\usepackage{balance}
\usepackage{tikz}
\usepackage{circuitikz}
\usetikzlibrary{arrows.meta}
\tikzstyle{block} = [draw, rectangle, 
    minimum height=3em, minimum width=6em]
\tikzstyle{sum} = [draw, circle, node distance=1cm]
\tikzstyle{input} = [coordinate]
\tikzstyle{output} = [coordinate]
\tikzstyle{pinstyle} = [pin edge={to-,thin,black}]

\graphicspath{{./figures/}}

\newtheorem{theorem}{Theorem}
\newtheorem{lemma}[theorem]{Lemma}
\newtheorem{corollary}{Corollary}[theorem]
\newtheorem{remark}{Remark}
\newtheorem{proposition}{Proposition}

\newcommand{\BBM}{\begin{bmatrix}}
\newcommand{\EBM}{\end{bmatrix}}
\newcommand{\BEQ}{\begin{equation}}
\newcommand{\EEQ}{\end{equation}}
\newcommand{\BIT}{\begin{itemize}}
\newcommand{\EIT}{\end{itemize}}

\newcommand{\st}{\mbox{subject to}}

\newcommand{\reals}{\mathbb{R}}

\newcommand{\minimize}{\mbox{minimize}}

\DeclareMathOperator{\Tr}{tr}
\DeclareMathOperator{\diag}{diag}
\DeclareMathOperator{\blkdiag}{blkdiag}

\DeclareMathOperator{\Expect}{\mathbb{E}}

\newcommand{\E}[2][]{\Expect_{#1}\!\left[\,#2\,\right]}   

\def\BibTeX{{\rm B\kern-.05em{\sc i\kern-.025em b}\kern-.08em
    T\kern-.1667em\lower.7ex\hbox{E}\kern-.125emX}}
\begin{document}
\title{Maximum Likelihood Estimation
for System Identification of Networks of Dynamical Systems}
\author{Anders Hansson, \IEEEmembership{Senior Member, IEEE}, João Victor Galvão da Mata, and Martin S.~Andersen
\thanks{This work was supported by ELLIIT, and by the Novo Nordisk Foundation under grant number NNF20OC0061894.}
\thanks{Anders Hansson is with the Department of Electrical Engineering, Linköping University (e-mail: anders.g.hansson@liu.se).}
\thanks{João Victor Galvão da Mata is with the Centre for Mathematical Sciences, Lund University, Sweden (e-mail: joao\_victor.galvao\_da\_mata@math.lu.se).}
\thanks{Martin S.~Andersen is with the Department of Applied Mathematics and Computer Science, Technical University of Denmark, Denmark (e-mail: mskan@dtu.dk).}}
\maketitle
\begin{abstract}
This paper considers maximum likelihood estimation of plant transfer
functions in a known interconnection of scalar autoregressive moving-average
models with exogenous inputs when only some network signals are measured. We
establish consistency of the estimated plant transfer functions and, under
additional local regularity and differential-identifiability conditions,
their asymptotic efficiency. We also derive an equivalent formulation that
avoids the \textcolor{black}{Riccati recursion} in the numerical optimization.
Numerical experiments with one and two recorded signals report estimation
accuracy, computation time, and convergence rates.
\end{abstract}
\subsubsection*{Keywords}
System Identification, Networks, Dynamical Systems, Maximum 
likelihood estimation, Consistency, Efficiency.
\section{Introduction}
\IEEEPARstart{S}{ystem} 
identification of linear dynamic networks has received considerable attention in recent years, motivated by the need to model large-scale interconnected systems in engineering 
and economics, e.g. \cite{VANDENHOF201823,linder2017indirect,Materassi_topology}.

Research on system identification of dynamic networks can broadly be divided into three main categories. The first category focuses on estimation of the network topology, possibly together with the dynamics, aiming to infer the interconnection structure from data, e.g. \cite{Materassi_topology,CHIUSO20121553,6125232}. The second category addresses identification of a single module in a network with known topology, e.g. 
\cite{VANDENHOF20132994,Dankers_predictor_input, Materassi2019SignalSF, 7402842, Linder03042017,9247487,9661422}, which can be seen as a generalization of classical closed-loop system identification \cite{GUSTAVSSON197759,FORSSELL19991215,vanDonkelaar1997tailor}. 
The third category concerns identification of the full network dynamics for a given topology, 
e.g. \cite{WEERTS2018256,FONKEN2022110295,hen+gev+baz19,Bombois2018optimal}. 

The present paper considers a known interconnection of scalar autoregressive
moving-average models with exogenous inputs (ARMAX) of prescribed orders.
\textcolor{black}{The exogenous signals are known. We consider partial node
measurements, meaning that only selected module inputs or outputs are
measured. Identifiability of the plant transfer functions depends on
which signals are measured and on the exogenous signals.} The quantities
to estimate are the plant transfer functions. The noise coefficients
are auxiliary and need not be identifiable under \textcolor{black}{partial
node measurements}.

We formulate a Gaussian maximum likelihood (ML) problem using the
\textcolor{black}{measured} signals. Its innovation-based criterion is
also a prediction-error method (PEM). We call it structured because
the same ARMAX coefficients determine the response to the
\textcolor{black}{exogenous signals} and the distribution of the prediction
errors. Under the stated
\textcolor{black}{conditions}, the plant estimates are consistent
\textcolor{black}{for both predictors when the time-invariant predictor is
initialized at zero}. Additional local
regularity \textcolor{black}{and differential identifiability} conditions
give asymptotic efficiency \textcolor{black}{for the plant parameters}.

\textcolor{black}{Partial node measurements produce} a coupled nonconvex
criterion. We derive an equivalent formulation \textcolor{black}{for zero
initial values} that \textcolor{black}{eliminates the unmeasured signals}
and avoids the \textcolor{black}{Riccati recursion} in the numerical problem
\cite{HANSSON20121955,Wallin02112014}.

The contributions are:
\begin{enumerate}
    \item a structured Gaussian likelihood and conditions for consistency
    \textcolor{black}{of the plant transfer function estimates and for
    asymptotic efficiency of the plant parameter estimates}.
    \item an equivalent formulation for numerical optimization
    \textcolor{black}{with zero initial values}.
\end{enumerate}
The consistency analysis extends \cite{14f6a003-64aa-3900-9a30-7d5cbc9eab0a}
to innovation-form network models and gives an alternative ergodic argument
to \cite{ljung_conver_1978}. Most results underlying Item~2 appeared in
\cite{10590797}.

The paper first derives the predictor and the statistical conditions, then
establishes consistency and efficiency, derives the equivalent numerical
formulation, and presents numerical results.
\subsection{Notation}
We consider all vectors to be column matrices. 
With $\hat e_i$ we denote the $i$th standard basis vector.
With $\mathbf 1$ we denote the vector of all ones. If we have
defined scalars or vectors $x_k$ for $1\leq k\leq n$, then we have also tacitly defined the 
vector $x=(x_1,\ldots,x_n)$. We denote by $A\otimes B$ the
Kronecker product
of $A$ with $B$.
We use $X\succ0$ and $X\succeq0$ to denote positive definiteness and positive semidefiniteness of the
symmetric matrix $X$, respectively. The notation $A\succ B$ is equivalent to $A-B\succ 0$, and
similarly for positive semidefiniteness. 
We let boldface letters denote the Z-transform of signals. These signals can be scalar-valued,
vector-valued or matrix-valued. We use the same notation for random variables as for realizations of 
random variables. With $\E{X}$ is meant the expected value of the random variable $X$.
\section{Model of Network}
We consider single-input single-output ARMAX models  
\BEQ\label{eqn:ARMAX-model}
y_k^i+\sum_{j=1}^{n^i} a_j^iy_{k-j}^i  = \sum_{j=1}^{n^i} b_j^i u_{k-j}^i + e_k^i + \sum_{j=1}^{n^i} c_j^i e_{k-j}^i
\EEQ
for $1\leq i\leq M$, 
where $e^i_k$ is a sequence of independent zero mean Gaussian random variables.  We assume
that these sequences are independent of one another for different $i$ and that they have the variances
$\lambda^i$.\footnote{Gaussianity is required for the exact likelihood and
the efficiency result. \textcolor{black}{Consistency can hold without Gaussianity.}}
We write $a^i=(a_1^i,\ldots,a_{n^i}^i)$ and we define $b^i$ and $c^i$ similarly. We also let 
$a=(a^1,\ldots,a^M)$ and we define $b$, $c$ and $\lambda$ similarly. 
We let $n=\sum_{i=1}^Mn^i$. 

We then define a network interconnection
of the ARMAX models via 
\BEQ
u_k=\Upsilon y_k+\Omega r_k\label{eqn:feedback}
\EEQ
where the exogenous signal $r_k$ and the zero-one matrices
$\Upsilon\in\reals^{M\times M}$ and $\Omega\in\reals^{M\times m}$ are known.\footnote{We assume that the closed loop system is well-posed.} We define the observed output as 
\BEQ
x_{o,k}=T_o\BBM y_k\\u_k\EBM\label{eqn:observed}
\EEQ
where the known matrix $T_o$ has full row rank and its $p$ rows are standard
basis vectors. Only $x_o$ and $r$ are available to the estimator; signals
used locally to implement \eqref{eqn:feedback} need not be centrally recorded.
The target is the collection of plant transfer functions, equivalently
$(a,b)$. The variables $(c,\lambda)$ specify the likelihood and are not
generally identifiable under partial recording.

For analysis, we use transfer functions and state-space descriptions. The
numerical formulation instead uses Toeplitz matrices to describe the
dynamical systems.

Define the polynomials
\begin{align*}
\mathbf A^i(z)&=z^{n^i}+a_1^iz^{n^i-1}+\cdots+a_{n^i}^i\\
\mathbf B^i(z)&=b_1^iz^{n^i-1}+\cdots+b_{n^i}^i\\
\mathbf C^i(z)&=z^{n^i}+c_1^iz^{n^i-1}+\cdots+c_{n^i}^i
\end{align*}
in the complex variable $z$.
We make the standing assumption that 
\begin{enumerate}
\item[A1.] $\mathbf A^i(z)$ and $\mathbf B^i(z)$ are coprime.
\end{enumerate}
We also define the transfer functions 
$$\mathbf G^i(z)=\frac{\mathbf B^i(z)}{\mathbf A^i(z)},\quad
\mathbf H^i(z)=\frac{\mathbf C^i(z)}{\mathbf A^i(z)}.$$
and $\mathbf G(z)=\blkdiag\left(\mathbf G^1(z),\ldots,\mathbf G^M(z)\right)$
and $\mathbf H(z)$
similarly. A model is defined as the vector $\mathbf M(z)=(\mathbf G(z),\mathbf H(z))$. 

For each ARMAX model we introduce an observer form realization
\begin{align*}
\xi_{k+1}^i&=F^i\xi_k^i+b^iu^i_k+(c^i-a^i)e_k^i\\
y^i_k&=H^i\xi^i_k+e_k^i,
\end{align*}
where $F^i= U-a^i\hat e_1^T$, $U$ is a shift matrix with ones on the first super-diagonal, and where 
$H^i=\hat e_1^T$. Let $\xi_k=(\xi^1_k,\ldots ,\xi^M_k)$, and define $u_k$, $y_k$ and $e_k$ similarly.
Then with 
\begin{align*}
F&=\blkdiag(F^1,\ldots,F^M)\\
B&=\blkdiag(b^1,\dots,b^M)\\
C&=\blkdiag(c^1-a^1,\dots,c^M-a^M)\\
H&=\blkdiag(H^1,\dots,H^M),
\end{align*}
it holds that 
\begin{align*}
\xi_{k+1}&=F\xi_k+Bu_k+Ce_k\\
y_k&=H\xi_k+e_k.
\end{align*}
The closed-loop system is by \eqref{eqn:feedback} given by
\begin{align*}
\xi_{k+1}&=F_c\xi_k+G_rr_k+G_ee_k\\
\BBM y_k\\u_k\EBM&=H_c\xi_k+J_rr_k+J_ee_k,
\end{align*}
where $F_c=F+B\Upsilon H,\;G_r=B\Omega,\; G_e=C+B\Upsilon$ and
\begin{align*}
H_c=\BBM I\\\Upsilon\EBM H,\quad
J_r=\BBM 0\\\Omega\EBM,\quad J_e=\BBM I\\\Upsilon\EBM.
\end{align*}
From \eqref{eqn:observed} we obtain the following state-space description of the observed outputs:
\BEQ\label{eqn:closed-loop}
\begin{aligned}
\xi_{k+1}&=F_c\xi_k+G_rr_k+G_ee_k\\
x_{o,k}&=H_o\xi_k+J_{ro}r_k+J_{eo}e_k,
\end{aligned}
\EEQ
where $H_o=T_oH_c$, $J_{ro}=T_oJ_r$ and $J_{eo}=T_oJ_e$.
We remark that $e_k$ has covariance $\Sigma_e=\diag(\lambda^1,\ldots,\lambda^M)$, which we tacitly 
assume is positive definite throughout this paper.  

We denote the 
closed loop transfer functions from $\mathbf r$ and $\mathbf e$ to $\mathbf x_o$ with $\mathbf G_c(z)$ and 
$\mathbf H_c(z)$, respectively. We also define $\mathbf M_c(z) 
=(\mathbf G_c(z),\mathbf H_c(z))$. We define the parameter vector
$\theta=(a,b,c,\lambda)$. \textcolor{black}{We denote by $\Theta_o$ the set
of parameters $\theta$ for which all eigenvalues of $F_c(\theta)$ lie
strictly inside the unit circle.}
We also define the compact set $\Theta\subset \Theta_o$ such that it contains the parameter
$\theta_0$ which represents the closed loop system from which we will collect the 
data $x_o$. It is over this set we will later on optimize a likelihood function. 
The model set $\mathcal M$ is the 
set of all models $\mathbf M_c(z) $ such that $\theta\in\Theta$. Sometimes we will 
explicitly write e.g. $\mathbf M_c(z;\theta)$ to emphasize the parameter 
dependence $\theta$
of the transfer function. This also goes for all other variables. 
We will sometimes tacitly assume that $\theta\in\Theta$. 
\section{Predictor}
To obtain a predictor that can be used for system identification,
we equivalently describe the observations using the innovation form. For the
closed-loop system in \eqref{eqn:closed-loop}, the steady-state Kalman
covariance identities give
\BEQ\label{eq:ARE}
\begin{aligned}
&\BBM I&K\\0&I\EBM\BBM \Sigma&0\\0&\Sigma_\epsilon\EBM\BBM I&K\\0&I\EBM^T\\&\qquad=
\BBM F_c&G_e\\H_o&J_{eo}\EBM\BBM \Sigma &0\\0&\Sigma_e\EBM\BBM F_c&G_e\\H_o&J_{eo}\EBM^T
\end{aligned}
\EEQ
in the variables $\Sigma$, $\Sigma_\epsilon$ and $K$. 
Solving for $K$ in the (1,2)-block  and substituting into the (1,1)-block  gives the
Riccati equation for the variable $\Sigma$. 
The (2,2)-block of \eqref{eq:ARE} defines
$\Sigma_\epsilon=H_o\Sigma H_o^T+J_{eo}\Sigma_eJ_{eo}^T$
which is the covariance for the 
innovations. The innovation form is then given by 
\BEQ\label{eqn:innovation}
\begin{aligned}
\hat \xi_{k+1}&=F_c\hat \xi_k+G_rr_k+K\epsilon_k\\
x_{o,k}&=H_o\hat \xi_k+J_{ro}r_k+\epsilon_k,
\end{aligned}
\EEQ
 where $\epsilon_k$ are the innovations, which
are such that $\epsilon_j$ is independent of $\epsilon_k$  whenever $j\neq k$.
Moreover, it is zero mean and Gaussian with covariance $\Sigma_\epsilon$, as already mentioned. The innovation form 
also defines the Kalman predictor 
$$\hat x_{o,k}=H_o\hat \xi_k+J_{ro}r_k,$$
and we have that the prediction error is the innovation $\epsilon_k=x_{o,k}-\hat x_{o,k}=
x_{o,k}-H_o\hat \xi_k-J_{ro}r_k$. From this we realize that the predictor is governed by the 
equations
\BEQ\label{eqn:time-invariant-predictor}
\begin{aligned}
\hat \xi_{k+1}&=F_c\hat \xi _k+G_rr_k+K\left(x_{o,k}-H_o\hat \xi _k-J_{ro}r_k\right)\\
\hat x_{o,k}&=H_o\hat \xi_k+J_{ro}r_k,
\end{aligned}
\EEQ
which is the Kalman filter. 
Therefore, the predictor only depends on $r_k$ and on what we observe, i.e., $x_{o,k}$. The
transfer function from $\mathbf x_o$ and $\mathbf r$ to $\hat{\mathbf x}_0$ is given by
$\mathbf W(z)=(\mathbf W_o(z),\mathbf W_r(z))$, where
\BEQ\label{eqn:predictor-transfer}
\begin{aligned}
\mathbf W_o(z)&=H_o(zI-F_c+KH_o)^{-1}K \\
\mathbf W_r(z)&=H_o(zI-F_c+KH_o)^{-1}(G_r-KJ_{ro})+J_{ro}. 
\end{aligned}
\EEQ
\begin{lemma}\label{lem:DARE}
Assume that $\theta\in\Theta$.
There is a unique solution $(\Sigma,K,\Sigma_\epsilon)$ to \eqref{eq:ARE} such that 
the matrix $F_c-KH_o$ has all its eigenvalues strictly inside the unit disc, and such that
$\Sigma_\epsilon$ is positive definite \textcolor{black}{if}
\begin{itemize}
\item[A2.] $J_{eo}$ has full row rank,
\item[A3.] $\mathbf C^i(z)\neq0$ for  $|z|=1$ and for all $i$.
\end{itemize}
\end{lemma}
\begin{IEEEproof}
There is a unique solution $(\Sigma,K,\Sigma_\epsilon)$ to \eqref{eq:ARE} such that 
the matrix $F_c-KH_o$ has all its eigenvalues strictly inside the unit disc, and such that
$\Sigma_\epsilon$ is positive definite if and only if $(F_c,H_o)$ is detectable and the pencil
$$\mathbf P(z)=\BBM -z I+F_c& G_e\\H_o&J_{eo}\EBM$$
has full row rank for $|z|=1$, see \cite[Theorem 2]{HANSSON1999245}. 
\textcolor{black}{Since $\theta\in\Theta\subset\Theta_o$, zero observer gain gives detectability.}
We can write the pencil as
$$\mathbf P(z)=\BBM I&0\\0&J_{eo}\EBM\BBM I&B\Upsilon\\0&I\EBM\BBM -z I+F& C\\H&I\EBM.$$
The first matrix has full row rank if and only if  $J_{eo}$ has full row rank.
The middle matrix is clearly invertible. The right matrix can after permutations of rows and
columns be seen to be a block diagonal matrix with blocks
\begin{align*}
\mathbf Q^i(z)&=\BBM -z I+F^i &c^i-a^i\\H^i& 1\EBM\\&=
\BBM 
-z -a_1^i& 1& 0&c_1^i-a_1^i\\
-a_2^i &-z & 1&c_2^i-a_2^i\\
-a_3^i& 0 & -z &c_3^i-a_3^i\\
1 & 0 & 0 & 1
\EBM
\end{align*}
for the case of $n^i=3$. We have with 
$$\mathbf U(z)=\BBM z^2&z &1&0\\0&1&0&0\\1&0&0&0\\0&0&0&1\EBM,\quad
\mathbf V(z)=\BBM 1&0&0&-1\\0&1&0&0\\0&0&1&0\\0&0&0&1\EBM$$
that
$$\mathbf U(z) \mathbf Q^i(z)\mathbf V(z)=
\BBM -\mathbf A^i(z)&0&0&\mathbf C^i(z)\\
-a_2^i&-z&1&c_2^i\\
-z-a_1^i&1&0&z+c_1^i\\
1& 0&0&0
\EBM.$$
We can, after this, multiply with further unitary matrices from the left and the right to see that the matrix
will be invertible for $|z|=1$ if and only if the 
polynomial  $\mathbf C^i(z)$ is not equal to zero for $|z|=1$.
\end{IEEEproof}
\begin{remark}\label{rem:row-rank}
Assumption A2 can always be fulfilled by removing redundant observed
signals. To see this, let $U$ be an  invertible matrix such that
$$\bar J_{eo}=\BBM \bar J_{eo1}\\0\EBM=U J_{eo},\quad 
\bar J_{ro}=\BBM \bar J_{ro1}\\\bar J_{ro2}\EBM=U J_{ro}$$
with $\bar J_{eo1}$ full row rank, and where the rows of 
$\bar J_{eo1}$ are a subset of the rows of $J_{eo}$. Then
$$\bar x_{o,k}=\BBM \bar x_{o1,k}\\\bar x_{o2,k}\EBM=
Ux_{o,k}=\BBM \bar J_{eo1}(H\xi_k+e_k)+\bar J_{ro1}r_k\\\bar J_{ro2}r_k\EBM. $$
The observation $\bar x_{o2,k}$ is not going to affect the prediction,
and hence can be discarded. This holds since after 
transformation with $U$ there will be rows in the parenthesis 
in \eqref{eqn:time-invariant-predictor} that are zero. Hence, a $K$ with
fewer columns is sufficient. 
\end{remark}
\begin{remark}
If $\Sigma_e$ is not diagonal or positive definite, we may just
consider new matrices defined by
$$\BBM G_e\\J_{eo}\EBM \Sigma_e\BBM G_e\\J_{eo}\EBM^T=
\BBM \bar G_e\\\bar J_{eo}\EBM \bar \Sigma_e\BBM \bar G_e\\\bar J_{eo}\EBM^T$$
where $\bar \Sigma_e$ is positive definite and diagonal, 
and then the matrices with bars instead of the matrices without bars should
be used throughout the paper.
Also, the results of Section~\ref{sec:equiv} only hold for diagonal and positive definite $\Sigma_e$. 
\end{remark}

{\color{black}To determine when \eqref{eq:ARE} has a stabilizing solution with
$\Sigma=0$, we use an equivalent formulation with zero cross covariance
between the transformed state disturbance and measurement disturbance. Let}
$$\BBM Q&S\\S^T&R\EBM=\BBM G_e\\J_{eo}\EBM \Sigma_e\BBM G_e\\J_{eo}\EBM^T. $$
We then have that \eqref{eq:ARE} is equivalent to 
\BEQ\label{eq:ARE-QRS}
\BBM I&K\\0&I\EBM\BBM \Sigma&0\\0&\Sigma_\epsilon\EBM\BBM I&K\\0&I\EBM^T
=\BBM F_c\\H_o\EBM\Sigma \BBM F_c\\H_o\EBM^T+\BBM Q&S\\S^T&R\EBM.
\EEQ
Multiply this equation with $\BBM I&-SR^{-1}\\0&I\EBM$
from the left and its transpose from the right to obtain 
\BEQ\label{eq:ARE-bar}
\begin{aligned}
\Sigma +\bar K\Sigma_\epsilon\bar K^T&=\bar F_c\Sigma \bar F_c^T+\bar Q\\
\bar K\Sigma_\epsilon&=\bar F_c\Sigma H_o^T\\
\Sigma_\epsilon&=H_o\Sigma H_o^T+R,
\end{aligned}
\EEQ
where $\bar F_c=F_c-SR^{-1}H_o$, $\bar Q=Q-SR^{-1}S^T$ and $\bar K=K-SR^{-1}$.
By solving for $\bar K$ in the second equation and substitution into
the first equation we obtain 
\BEQ\label{eq:ARE-nocross}
\Sigma =\bar F_c\Sigma \bar F_c^T+\bar Q-
\bar F_c\Sigma H_o^T(H_o\Sigma H_o^T+R)^{-1}H_o\Sigma \bar F_c^T.
\EEQ
\textcolor{black}{Setting $\Sigma=0$ in \eqref{eq:ARE-nocross} gives $\bar Q=0$.}

\begin{theorem}\label{thm:trivial-sol}
Assume that $\theta\in\Theta$, that Assumption A2 holds and that
$\mathbf C^i(z)\neq 0$ for $|z|\geq 1$ and for all $i$.
If 
\BEQ\label{eqn:Sigma_e}
\Sigma_e-\Sigma_eJ_{eo}^TR^{-1}J_{eo}\Sigma_e=0,
\EEQ
then $(\Sigma,K,\Sigma_\epsilon)=
(0,SR^{-1},R)$ is a 
solution to \eqref{eq:ARE} such that  $F_c-KH_o$ has all eigenvalues
strictly inside the unit circle. 
\end{theorem}
\begin{IEEEproof}
Under the assumption we have that $\bar Q=Q-SR^{-1}S^T=0$, from 
which is straightforward to see that $(\Sigma,K,\Sigma_\epsilon)=
(0,SR^{-1},R)$ satisfies the equivalent formulation of the algebraic
Riccati equation in \eqref{eq:ARE-bar}. 
It holds that 
\begin{align*}
\bar F_c&=F_c-KH_o=F+B\Upsilon H-SR^{-1}H_o\\
&=F+\left(B\Upsilon -G_e\Sigma_eJ_{eo}^TR^{-1}J_{eo}\right)H\\
&=F+(B\Upsilon -B\Upsilon -C)H=F-CH,
\end{align*}
where we have used the definition of $G_e$ and the fact that 
$J_{eo}^TR^{-1}J_{eo}=\Sigma_e^{-1}$
by our assumption. We see that $\bar F_c$ is a block diagonal matrix
with blocks $U-a^i\hat e_1^T-(c^i-a^i)\hat e_1^T=U-c^i\hat e_1^T$, which are on observer
form. From our assumption on $\mathbf C^i$ it now follows that all blocks
have their eigenvalues strictly inside the unit circle. 
\end{IEEEproof}
\begin{remark}\label{rem:all-and-one}
{\color{black}Under \eqref{eqn:Sigma_e}, $J_{eo}^TR^{-1}J_{eo}=\Sigma_e^{-1}$,
so $J_{eo}$ has full column rank and is invertible by A2. By
\eqref{eqn:feedback} and \eqref{eqn:observed},
$\mathbf y=J_{eo}^{-1}(\mathbf x_o-J_{ro}\mathbf r)$.
Since $\Sigma=0$, the state prediction error is zero in mean square,
and \eqref{eqn:closed-loop} and \eqref{eqn:innovation} give
$\epsilon_k=J_{eo}e_k$. Using \eqref{eqn:ARMAX-model} and
\eqref{eqn:feedback}, the predictor is
$$\hat{\mathbf x}_o=\mathbf x_o-J_{eo}\mathbf H^{-1}
\bigl(\mathbf y-\mathbf G(\Upsilon\mathbf y+\Omega\mathbf r)\bigr).$$
For fixed $\mathbf H$, this expression is affine in $\mathbf G$.}

One example of when \eqref{eqn:Sigma_e} is fulfilled is when 
$T_o=\BBM I&0\EBM$, that is when $x_{o,k}=y_k$. 
\end{remark}
For \textcolor{black}{output-error (OE)} models, i.e. when $a=c$, we have the following theorem.
\begin{theorem}\label{thm:trivial-sol-OE}
Assume that $\theta\in\Theta$, that Assumption A2 holds, that
$\mathbf A^i(z)\neq 0$ for $|z|\geq 1$ for all $i$, and
that $a=c$. If
\BEQ\label{eqn:lin-condition}
\Upsilon\left(I-J_{eo}^T\left(J_{eo}J_{eo}^T\right)^{-1}J_{eo}\right)=0
\EEQ
then $(\Sigma,K,\Sigma_\epsilon)=(0,SR^{-1},R)$ is a 
solution to \eqref{eq:ARE} such that $F_c-KH_o$ has all eigenvalues inside
the unit circle. 
\end{theorem}
\begin{IEEEproof}
{\color{black}By A2, $J_{eo}J_{eo}^T$ is invertible. Define
$$L_\Upsilon=\Upsilon J_{eo}^T(J_{eo}J_{eo}^T)^{-1}.$$
Equation \eqref{eqn:lin-condition} gives $\Upsilon=L_\Upsilon J_{eo}$.
Since $R=J_{eo}\Sigma_eJ_{eo}^T$, we obtain
$$\Upsilon\Sigma_eJ_{eo}^TR^{-1}
=L_\Upsilon R R^{-1}=L_\Upsilon.$$
With $a=c$, we have $G_e=B\Upsilon$ and hence}
\begin{align*}
\bar Q&=B\Upsilon\left(\Sigma_e-
\Sigma_eJ_{eo}^TR^{-1}J_{eo}\Sigma_e\right)\Upsilon^TB^T\\
&={\color{black}B\left(\Upsilon-L_\Upsilon J_{eo}\right)\Sigma_e\Upsilon^TB^T=0.}
\end{align*}
{\color{black}For $(\Sigma,K,\Sigma_\epsilon)=(0,SR^{-1},R)$, we have $\bar K=0$.
The three equations in \eqref{eq:ARE-bar} reduce to $0=\bar Q$, $0=0$,
and $\Sigma_\epsilon=R$, respectively. Also,}
\begin{align*}
F_c-KH_o&=F+B\Upsilon H-B\Upsilon \Sigma_eJ_{eo}^TR^{-1}J_{eo}H\\
&={\color{black}F+B\left(\Upsilon-L_\Upsilon J_{eo}\right)H=F.}
\end{align*}
{\color{black}Each block $F^i$ of $F$ has characteristic polynomial $\mathbf A^i(z)$.
By the assumption on $\mathbf A^i(z)$, all eigenvalues of $F$ lie strictly
inside the unit circle.}
\end{IEEEproof}
\begin{corollary}\label{cor:predictor-OE}
Under the assumptions of Theorem~\ref{thm:trivial-sol-OE}, \textcolor{black}{including
\eqref{eqn:lin-condition},} the predictor transfer functions in
\eqref{eqn:predictor-transfer} are given by
\BEQ\label{eqn:predictor-transfer-OE}
\begin{aligned}
\mathbf W_o(z)&=H_o(zI-F)^{-1}K \\
\mathbf W_r(z)&=H_o(zI-F)^{-1}(G_r-KJ_{ro})+J_{ro}
\end{aligned}
\EEQ
where \textcolor{black}{$K=BL_\Upsilon$ and $G_r-KJ_{ro}=B(\Omega-L_\Upsilon J_{ro})$.}
\end{corollary}
\begin{IEEEproof}
{\color{black}Substitute $F_c-KH_o=F$ from the proof of Theorem~\ref{thm:trivial-sol-OE} into
\eqref{eqn:predictor-transfer} to obtain \eqref{eqn:predictor-transfer-OE}.
The proof of the theorem gives
$$K=SR^{-1}=B\Upsilon\Sigma_eJ_{eo}^TR^{-1}=BL_\Upsilon.$$
Since $G_r=B\Omega$, we have $G_r-KJ_{ro}=B(\Omega-L_\Upsilon J_{ro})$.}
\end{IEEEproof}
\begin{remark}\label{rem:affine-transformations}
{\color{black}In \eqref{eqn:predictor-transfer-OE},
$H_o(zI-F)^{-1}B=J_{eo}\mathbf G(z)$. The factors $L_\Upsilon$ and
$\Omega-L_\Upsilon J_{ro}$ are independent of $\theta$, so the predictor
transfer functions are affine in $\mathbf G(z)$.
PEM criteria are separable when the objective splits into terms in
disjoint parameter groups that can be minimized independently.}
\end{remark}
\section{Informative Signals}
A crucial assumption needed in order to identify a model of a dynamical system
is the informativity of the signals collected for system identification. 
The signals $(x_{o},r)$ are said to be {\it informative enough} if for any two models $\mathbf M_c(z;\theta_1)$ and $\mathbf M_c(z;\theta_2)$ in $\mathcal M$, it holds that 
\BEQ\label{eqn:informative}
\E{\left(\hat x_{o,k}(\theta_1)-\hat x_{o,k}(\theta_2)\right)
\left(\hat x_{o,k}(\theta_1)-\hat x_{o,k}(\theta_2)\right)^T}=0
\EEQ
implies that $\mathbf W\left(e^{i\omega};\theta_1\right)=\mathbf W\left(e^{i\omega};\theta_2\right)$ for almost  
all $\omega$. 
The expectation is here with respect to the innovation 
$\epsilon_k(\theta_0)$ as well as the 
\textcolor{black}{exogenous signal $r_k$. For almost every fixed realization of
stationary and ergodic $r$ independent of $e$, this expectation has the
time-average interpretation of Ljung \cite[Section~2.3, Eq.~(2.60)]{lju87}.}
For $i\in\{1,2\}$, let $\hat x_{o,k}(\theta_i)$ be the predictor output
defined by \eqref{eqn:time-invariant-predictor}, and let
$\mathbf W^i(z)=\mathbf W(z;\theta_i)$, with $\mathbf W$ defined by
\eqref{eqn:predictor-transfer}. With boldface letters denoting Z-transforms,
it holds that
$$\hat{\mathbf{x}}_{o}^2
-\hat{\mathbf{x}}_{o}^1=\Delta\mathbf W\mathbf z$$
where $\Delta\mathbf W=\mathbf W^2-\mathbf W^1$ and $\mathbf z=(\mathbf x_o,
\mathbf r)$, the latter being a function of $\theta_0$. 
From Parseval's formula it follows that 
\eqref{eqn:informative} is equivalent to 
$$\Delta \mathbf W\left(e^{i\omega}\right)\Phi_z(\omega)
\Delta \mathbf W\left(e^{-i\omega}\right)^T=0$$
where $\Phi_z$ is the spectrum of $z_k$. Hence, a sufficient condition for the 
above equality to imply that $\mathbf W\left(e^{i\omega};\theta_1\right)=\mathbf W\left(e^{i\omega};\theta_2\right)$ for almost  
all $\omega$, is that $\Phi_z(\omega)$ is positive definite for almost all 
$\omega$. From the innovation form it follows that 
a realization of the transfer function $\mathcal H(z)$ from $(\mathbf r,\boldsymbol\epsilon)$
to $\mathbf z$ is given by
$$(\mathcal A,\mathcal B,\mathcal C, \mathcal D)=\left(F_c,\BBM G_r&K\EBM,\BBM H_o\\0\EBM,\BBM J_{ro}&I\\I&0\EBM\right).$$
It holds that 
\begin{align*}
&\BBM I&0\\\mathcal C(zI-\mathcal A)^{-1}&I\EBM \mathcal M(z)\BBM I&(zI-\mathcal A)^{-1}\mathcal B\\0&I\EBM\\&=
\BBM -zI+\mathcal A&0\\0& \mathcal H(z)\EBM,
\end{align*}
where
$$\mathcal M(z)=\BBM -zI+\mathcal A&\mathcal B\\\mathcal C&\mathcal D\EBM.$$
\textcolor{black}{Since $\mathcal A=F_c$ and $\theta_0\in\Theta_o$,
$-zI+\mathcal A$ is invertible on the unit circle.} Hence
$\mathcal H(z)$ has full 
rank on the unit circle if and only if $\mathcal M(z)$ has
full rank on the unit circle.
The rank of $\mathcal M(z)$ is full if and only if the rank of 
$$\BBM -zI+F_c&K\\H_o&I\EBM=
\BBM I&K\\0&I\EBM\BBM -zI+F_c-KH_o&0\\H_o&I\EBM$$
is full. 
This matrix drops in rank
for the eigenvalues of $F_c-KH_o$.  
Since the eigenvalues by Lemma~1 are
strictly inside the unit circle the rank is actually full on the unit circle. 
Therefore the rank of $\mathcal H(z)$ is full on the unit circle, and hence
$\Phi_z(\omega)$ is positive definite for all $\omega$ if and only if
the spectrum for
$(r_k,\epsilon_l)$ is positive definite for all $\omega$.
We now assume that the \textcolor{black}{exogenous signal $r$} and
\textcolor{black}{the disturbance process $e$} are independent. The experiment remains a closed-loop network experiment because
$y_k$ and $u_k$ are interconnected by \eqref{eqn:feedback}. It then follows
that also $r_k$ and $\epsilon_k$ are independent of one another. Since
$\Sigma_e$ is assumed to be positive definite, we introduce the assumption 
\begin{enumerate}
\item[A4.] $\Phi_r(\omega)\succ 0,\quad\forall\omega$, and
\textcolor{black}{the process $r$ is stationary and ergodic and independent of the process $e$},
\end{enumerate}
which together with assumptions A2 and A3 will guarantee that the signals
$(x_o,r)$ are informative enough. If $m=1$, one scalar excitation can satisfy
A4. Recovery of every plant module additionally requires A5
\textcolor{black}{(stated in Section~\ref{sec:identifiability})} and depends on
$\Omega$ and $T_o$ \cite{Bazanella2010closed}. We summarize the informativity
result in a lemma.
\begin{lemma}\label{lem:informative}
Assume that Assumptions A2--A4 hold. Then the signals $(x_o,r)$ are informative
enough, i.e. for any two models $\mathbf M_c(z;\theta_1)$ and $\mathbf M_c(z;\theta_2)$ in $\mathcal M$, it holds that \eqref{eqn:informative}
implies that $\mathbf W\left(e^{i\omega};\theta_1\right)=\mathbf W\left(e^{i\omega};\theta_2\right)$ for almost  
all $\omega$. 
\end{lemma}
\section{Identifiability}\label{sec:identifiability}
We first recover the innovation-form transfer matrices from $\mathbf W$ and
then use network identifiability to recover $\mathbf G$. Let
\BEQ\label{eqn:Go}
\mathbf G_o(z)=H_o(zI-F_c)^{-1}K+I.
\EEQ
With the definition in \eqref{eqn:Go}, applying the Z-transform to the
innovation form gives
\BEQ\label{eqn:innovation-output}
\begin{aligned}
\mathbf x_o&=\mathbf G_c\mathbf r+\mathbf G_o\boldsymbol\epsilon,\\
\hat{\mathbf x}_o&=\mathbf G_c\mathbf r+(\mathbf G_o-I)\boldsymbol\epsilon,\\
\mathbf x_o&=\hat{\mathbf x}_o+\boldsymbol\epsilon.
\end{aligned}
\EEQ
Since $F_c-KH_o$ is stable, the monic transfer matrix $\mathbf G_o$
has a stable inverse. In fact,
\begin{align*}
\mathbf G_o^{-1}
&=I-H_o(zI-F_c+KH_o)^{-1}K=I-\mathbf W_o.
\end{align*}
It follows from \eqref{eqn:innovation-output} that
\begin{align*}
\boldsymbol\epsilon
&=\mathbf G_o^{-1}(\mathbf x_o-\mathbf G_c\mathbf r),\\
\hat{\mathbf x}_o
&=\mathbf x_o-\boldsymbol\epsilon
=(I-\mathbf G_o^{-1})\mathbf x_o+\mathbf G_o^{-1}\mathbf G_c\mathbf r.
\end{align*}
Therefore, the predictor transfer matrices satisfy
\BEQ\label{eqn:innovation-to-predictor}
\begin{aligned}
\mathbf W_o&=I-\mathbf G_o^{-1},\\
\mathbf W_r&=\mathbf G_o^{-1}\mathbf G_c.
\end{aligned}
\EEQ
Solving \eqref{eqn:innovation-to-predictor}
for the innovation-form transfer matrices gives
\BEQ\label{eqn:predtoinnovation}
\begin{aligned}
\mathbf G_o&=(I-\mathbf W_o)^{-1}\\
\mathbf G_c&=(I-\mathbf W_o)^{-1}\mathbf W_r.
\end{aligned}
\EEQ
Notice also that
\begin{align*}
\mathbf W_o(I-\mathbf W_o)^{-1}=\mathbf G_o-I,
\end{align*}
which is the strictly proper part of the innovation transfer matrix.
From this, we realize that we can recover the innovation form transfer functions
for the closed
loop system from $\mathbf W$. We summarize this result in a lemma.
\begin{lemma}\label{lem:predtoinnovation}
Given the predictor transfer functions $\mathbf W$, the monic innovation transfer matrix $\mathbf G_o$ and
the closed-loop transfer matrix $\mathbf G_c$ are uniquely recovered using
\eqref{eqn:predtoinnovation}.
\end{lemma}
The closed-loop transfer matrix from $\mathbf r$ to $\mathbf x_o$ is
\BEQ\label{eqn:Gc}
\mathbf G_c=J_{eo}(I-\mathbf G\Upsilon)^{-1}\mathbf G\Omega +
J_{ro}.
\EEQ
Writing $\mathbf G_c(\mathbf G)$ for the right-hand side of
\eqref{eqn:Gc}, the true network is identifiable if
\begin{equation}\label{eqn:identifiability-at-G0}
\mathbf G_c(\mathbf G)=\mathbf G_c(\mathbf G_0)
\quad\Longrightarrow\quad
\mathbf G=\mathbf G_0.
\end{equation}
A measurement and excitation configuration is generically network
identifiable if this holds for almost all $\mathbf G_0$ consistent with the
network structure; see \cite{hen+gev+baz19}. We use the following assumption.
\begin{enumerate}
\item[A5.] The true network $\mathbf G_0$ satisfies
\eqref{eqn:identifiability-at-G0}.
\end{enumerate}
Thus, when A5 is established by a generic result, $\mathbf G_0$ is assumed
not to be exceptional. The paper assumes A5 rather than providing a general
test; its validity depends on the measured and excited nodes.

Consider the example of a network of dynamical systems as described by the block diagram in 
Figure~\ref{fig:dyn-network} which is taken from \cite{hen+gev+baz19}. 
Assume that we only measure $\mathbf u^3$. Some calculations show that 
\BEQ\label{eq:example-closed-loop}
\begin{aligned}
(1-\mathbf G^1\mathbf G^3)\mathbf u^3&=\mathbf G^1\mathbf r^1+\mathbf G^1\mathbf G^2\mathbf r^2+\mathbf r^3\\&+
\mathbf H^1\mathbf e^1+\mathbf G^1\mathbf H^2\mathbf e^2 + \mathbf G^1\mathbf H^3\mathbf e^3,
\end{aligned}
\EEQ
and from this we realize that \eqref{eqn:Gc} may be written as
$$\BBM \mathbf G^1&\mathbf G^1\mathbf G^2& 1\EBM - (1-\mathbf G^1\mathbf G^3)\mathbf G_c=0.$$
From this, we may derive that 
$$\mathbf G^1=\frac{\mathbf G_c^1}{\mathbf G_c^3},\quad
\mathbf G^2=\frac{\mathbf G_c^2}{\mathbf G_c^1},\quad
\mathbf G^3=\frac{\mathbf G_c^3-1}{\mathbf G_c^1}.$$
The formulas uniquely recover the modules whenever their denominators do not
vanish, establishing generic network identifiability. This algebraic property
is independent of whether $\mathbf G_c$ is first estimated indirectly or is
parametrized directly in terms of $\mathbf G$.
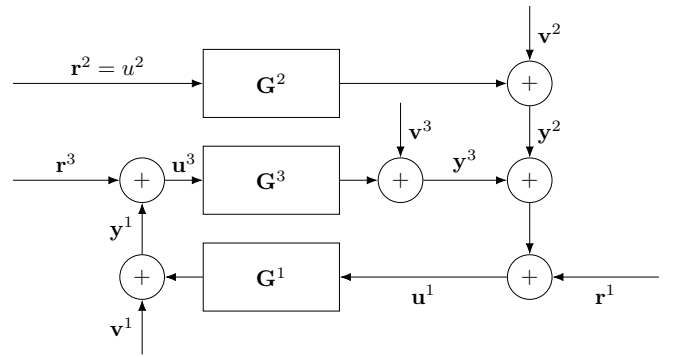
\begin{figure}[htbp]
\begin{center}
\resizebox{0.48\textwidth}{!}{
\begin{tikzpicture}[auto, node distance=2cm]
    \node [input, name=r2] {};
    \node [block, right of=r2, node distance=4cm] (G2) {$\mathbf G^2$};
    \node [sum, right of=G2, node distance=4cm] (sum2) {$+$};
    \node [input, above of=sum2, node distance=1.2cm] (e2) {};
    \node [block, below of=G2, node distance=1.5cm] (G3) {$\mathbf G^3$};
    \node [sum, right of=G3, node distance=2cm] (sum3) {$+$};
    \node [sum, right of=sum3, node distance=2cm] (sum23) {$+$};
    \node [input, above of=sum3, node distance=1.2cm] (e3) {};
    \node [sum, below of=sum23, node distance=1.5cm] (sum4) {$+$};
    \node [sum, left of=G3, node distance=2cm] (sum13) {$+$};
    \node [block, below of=G3, node distance=1.5cm] (G1) {$\mathbf G^1$};
    \node [sum, left of=G1, node distance=2cm] (sum1) {$+$};
    \node [input, below of=sum1, node distance=1.2cm] (e1) {};
    \node [input, left of=sum13] (r3) {};
    \node [input, right of=sum4] (r1) {};
    \draw [-Latex] (r1) -- node {$\mathbf r^1$} (sum4);        
    \draw [-Latex] (r2) -- node {$\mathbf r^2=u^2$} (G2);    
    \draw [-Latex] (r3) -- node {$\mathbf r^3$} (sum13);
    \draw [-Latex] (e1) -- node {$\mathbf v^1$} (sum1);
    \draw [-Latex] (e2) -- node {$\mathbf v^2$} (sum2);
    \draw [-Latex] (e3) -- node {$\mathbf v^3$} (sum3);
    \draw [-Latex] (G2) -- node {} (sum2);
    \draw [-Latex] (sum2) -- node {$\mathbf y^2$} (sum23);
    \draw [-Latex] (G3) -- node {} (sum3);
    \draw [-Latex] (sum3) -- node {$\mathbf y^3$} (sum23);
    \draw [-Latex] (sum23) -- node {} (sum4);
    \draw [-Latex] (sum4) -- node {$\mathbf u^1$} (G1);
    \draw [-Latex] (G1) -- node {} (sum1);
    \draw [-Latex] (sum1) -- node {$\mathbf y^1$} (sum13);
    \draw [-Latex] (sum13) -- node {$\mathbf u^3$} (G3);
\end{tikzpicture}}
\end{center}
\caption{Block diagram for a dynamic network where $\mathbf{v}^i
=\mathbf{H}^i\mathbf{e}^i$.}
\centering
\label{fig:dyn-network}
\end{figure}

The observed mapping does not generally determine $\mathbf H$ or
$\mathbf H_c$, as the disturbances are unobserved. Consequently, the main
consistency result concerns $(a,b)$ rather than $(c,\lambda)$; special cases
are discussed in Remark~\ref{rem:innovation-form}.

\section{Maximum Likelihood Estimation}\label{sec:maximum-likelihood}
There are two possible ways of defining
an ML problem for system identification. Already in \cite{ASTROM196596} 
it is mentioned that one possible formulation is to use the
innovation form in \eqref{eqn:innovation}, and make use of the fact that the
innovations $\epsilon_k$ are independent random variables with zero mean and covariance $\Sigma_\epsilon$.  Hence the 
probability density function (pdf)
of the innovations over the time interval one to $N$ is just
the product of the pdfs for each time instant, and 
by considering the negative logarithm of this function we obtain the 
optimization problem 
\BEQ\label{opt:ML-stationary}
\begin{aligned}
\minimize\quad&\frac{1}{2}\sum_{k=1}^N \left(\epsilon_k^T\Sigma_{\epsilon}^{-1}\epsilon_k+
\ln \det \Sigma_{\epsilon}\right)\\
\st\quad& \hat\xi_{k+1}=F_c \hat\xi_k+G_rr_k+K\epsilon_k\\
&x_{o,k}=H_o \hat \xi_k+J_{ro}r_k+\epsilon_k\\
&\eqref{eq:ARE},\quad\theta\in\Theta
\end{aligned}
\EEQ
with variables $(\theta, \hat \xi,\epsilon,\Sigma,K,\Sigma_\epsilon)$, and where 
$\xi_1$ is a free variable. Notice that $\Sigma_{\epsilon}$ 
is invertible by Lemma~\ref{lem:DARE} and that the algebraic Riccati equation 
in \eqref{eq:ARE} is
also a constraint in the above optimization problem. 
The innovations are functions of the observations 
$x_o=
(x_{o,1},\ldots,x_{o,N})$ through the 
constraints. Since $\hat\xi_1$ is a free variable, it can be argued that this
is not a true ML formulation. 

An alternative formulation of the ML problem is based on the fact that the pdf for the observations
can be factorized as 
$$p(x_o)=
\prod_{k=1}^Np_k(x_{o,k}|x_{o,k-1},\ldots,x_{o,1}).$$
Let us introduce the time-varying Kalman filter
\begin{align*}
\check\xi_{k+1}&=F_c\check\xi_k+G_rr_k+K_k(x_{o,k}-H_o\check\xi_k-J_{ro}r_k)\\
\check x_{o,k}&=H_o\check \xi_k+J_{ro}r_k
\end{align*}
where $K_k$ is obtained from the recursion
\BEQ\label{eq:Rrec}
\begin{aligned}
&\BBM I&K_k\\0&I\EBM\BBM \Sigma_{k+1}&0\\0&\Sigma_{\epsilon,k}\EBM\BBM I&K_k\\0&I\EBM^T\\&\qquad=
\BBM F_c&G_e\\H_o&J_{eo}\EBM\BBM \Sigma_k &0\\0&\Sigma_e\EBM\BBM F_c&G_e\\H_o&J_{eo}\EBM^T,
\end{aligned}
\EEQ
with $\check \xi_1$ and $\Sigma_1$ given, which are the expected value and 
covariance of $\xi_1$, respectively. 

If we assume that $r_k$ has been zero for $k\leq 0$, then
$\check\xi_1=0$ and $\Sigma_1= P$, where $P$ is the solution of the 
algebraic Lyapunov equation
\BEQ\label{eqn:Lyapunov}
P=F_cPF_c^T+G_e\Sigma_eG_e^T.
\EEQ
Associated with the time-varying Kalman-filter are the non-stationary innovations 
$\check\epsilon_k=x_{o,k}-\check x_{o,k}$, and the non-stationary innovation-form
\begin{align*}
\check\xi_{k+1}&=F_c\check\xi_k+G_rr_k+K_k\check\epsilon_k\\
x_{o,k}&=H_o\check \xi_k+J_{ro}r_k+\check\epsilon_k,
\end{align*}
with $\check\xi_1=0$, and 
where $\check\epsilon_k$ is a sequence of independent zero mean random variables with 
covariances $\Sigma_{\epsilon,k}$. It now holds that $p_k(x_{o,k}|x_{o,k-1},\ldots,x_{o,1})$
are pdfs for Gaussian random variables with mean $\check x_{o,k}$ and covariance 
$\Sigma_{\epsilon,k}$. By taking the logarithm of the factorized 
version of the pdf for the observations we see that the ML estimation problem is 
equivalent to the optimization problem 
\BEQ\label{opt:ML-timevaying}
\begin{aligned}
\minimize\quad&\frac{1}{2}\sum_{k=1}^N \left(\check\epsilon_k^T\Sigma_{\epsilon,k}^{-1}\check\epsilon_k+
\ln \det \Sigma_{\epsilon,k}\right)\\
\st\quad&\check \xi_{k+1}=F_c\check \xi_k+G_rr_k+K_k\check\epsilon_k\\
&x_{o,k}=H_o\check  \xi_k+J_{ro}r_k+\check\epsilon_k\\
&\eqref{eq:Rrec},\quad \eqref{eqn:Lyapunov},\quad\theta\in\Theta
\end{aligned}
\EEQ
with variables $(\theta, \check\xi,\check\epsilon,P,\Sigma_k,K_k,\Sigma_{\epsilon,k})$, and where 
$\check\xi_1=0$. Notice that $\Sigma_{\epsilon,k}$ is invertible by
Lemma~\ref{lem:bound-cov}. Both the algebraic Lyapunov equation
in \eqref{eqn:Lyapunov} as well as the Riccati recursion in \eqref{eq:Rrec}
are constraints in the optimization problem. 

Neither of the above formulations are very good formulations of the 
system identification problem for
numerical optimization, since the constraints depend on the parameter 
$\theta$ via the solution to either an
algebraic Lyapunov equation or an 
algebraic Riccati equation.
\textcolor{black}{For numerical optimization we present an
equivalent formulation of \eqref{opt:ML-timevaying} with zero initial values
that eliminates the Riccati recursion.}
\section{Consistency}
Consistency of ML estimation for system identification has a long history. The first investigation
was performed in \cite{ASTROM196596}, where single-input single-output ARMAX models were 
investigated.
Later references are e.g.
\cite{ljung_conver_1978,14f6a003-64aa-3900-9a30-7d5cbc9eab0a} and the 
references therein. The first reference of these two references shows convergence 
with probability one of parameters for linear, possibly time-varying, systems---hence not necessarily ergodic.
It considers general time-varying predictors. 
The second reference considers multiple-input multiple-output ARMAX
models under ergodic assumptions. It does in detail investigate
the time-varying predictor based on the conditional mean 
with respect to the observed data.
We will 
investigate consistency for the time-varying predictor and
obtain consistency for the time-invariant predictor for free. Our
proof is a generalization of the results in~\cite{14f6a003-64aa-3900-9a30-7d5cbc9eab0a}.
Introduce the log-likelihood function
$$L_N(\theta)=\frac{1}{2N}\sum_{k=1}^N \left(\check\epsilon_k^T\Sigma_{\epsilon,k}^{-1}\check\epsilon_k+
\ln \det \Sigma_{\epsilon,k}\right)+\frac{n_o}{2}\ln(2\pi),$$
where $\check \epsilon_k$ was defined in the previous section. 
Notice that both $\check\epsilon_k$ and $\Sigma_{\epsilon,k}$ are functions of $\theta$, i.e. we
may write $\check\epsilon_k(\theta)=x_{o,k}-\check x_{o,k}(\theta)$ and
$\Sigma_{\epsilon,k}(\theta)$.
The observation $x_{o,k}$ is generated by the true parameter $\theta_0$,
whereas $\check\epsilon_k(\theta)$ is evaluated at the candidate parameter
$\theta$.
We denote with $\hat \theta_N$
the optimizing argument of $L_N(\theta)$.  We want to show that $(\hat a_N,\hat b_N)\rightarrow (a_0,b_0)$ as 
$N\rightarrow \infty$ with probability one, i.e., we want to show what is called  {\it consistency}. 

The challenging part of proving consistency is to show that 
$$\sup_{\theta\in\Theta}\left|L_N(\theta)-L(\theta)\right|\rightarrow 0,\quad 
N\rightarrow \infty\;\mathrm{w.p. 1}$$
where $L(\theta)=-\E{\ln f(\epsilon_1(\theta);\theta)}$, and where
$$
f(\epsilon;\theta)=\frac{1}{\sqrt{(2\pi)^{n_0}\det \Sigma_\epsilon(\theta)}}
\exp\left(-\frac{1}{2}\epsilon^T\Sigma_\epsilon(\theta)^{-1}
\epsilon\right).$$
Notice that the time index 1 in $\epsilon_1(\theta)$ is arbitrary since 
$\epsilon_k(\theta)$ is a stationary stochastic process.
Denote by $\Theta_{\min}$
the set of $\theta\in\Theta$ that minimizes $L(\theta)$.
It then follows, c.f. \cite[Theorem 8.2]{lju87}, that 
$$\inf_{\theta\in\Theta_{\min}}\left\|\hat \theta_N-\theta\right\|_2
\rightarrow 0,\quad N\rightarrow \infty\; \mathrm{w.p.\,\, 1}.$$
We want to show that the set $\Theta_{\min}$ contains $\theta_0$. 
Notice that
$$L(\theta)=\frac{1}{2}\Tr Q\Sigma_\epsilon(\theta)^{-1}+
\frac{1}{2}\ln\det\Sigma_\epsilon(\theta)+\text{constant},$$
where $Q=\E{\epsilon_1(\theta)\epsilon_1(\theta)^T}$. 
\textcolor{black}{For fixed $Q$,} the function $L$ is a strictly convex function in $\Sigma_\epsilon(\theta)^{-1}$, \cite{han+and23}. 
The gradient of the function with respect to $\Sigma_\epsilon(\theta)^{-1}$ is
\textcolor{black}{$\tfrac{1}{2}(Q-\Sigma_\epsilon(\theta))$}
from which we conclude that $\Sigma_\epsilon(\theta)=Q$ is the unique minimizer of $L(\theta)$ with minimal value
$${\color{black}\tfrac{1}{2}\bigl(}\ln\det\E{\epsilon_1(\theta)\epsilon_1(\theta)^T}+n_o{\color{black}\bigr)}+\text{constant}.$$
From the fact that the mean square error is minimized by the 
Kalman filter, it holds that 
\BEQ\label{eqn:MSE}
\E{\epsilon_1(\theta)\epsilon_1(\theta)^T}\succeq 
\E{\epsilon_1(\theta_0)\epsilon_1(\theta_0)^T}=\Sigma_\epsilon(\theta_0)
\EEQ
for all $\theta\in\Theta$. Therefore, we have
\begin{align*}
L(\theta)&\geq {\color{black}\tfrac{1}{2}\bigl(}\ln\det\E{\epsilon_1(\theta)\epsilon_1(\theta)^T}+n_o{\color{black}\bigr)}+\text{constant}\\
&\geq L(\theta_0)={\color{black}\tfrac{1}{2}\bigl(}\ln\det\E{\epsilon_1(\theta_0)\epsilon_1(\theta_0)^T}{\color{black}+n_o\bigr)}
+\text{constant},
\end{align*}
where the second inequality follows from \eqref{eqn:MSE}. Equality between
$L(\theta)$ and $L(\theta_0)$ holds only if 
$$\E{\epsilon_1(\theta)\epsilon_1(\theta)^T}=
\E{\epsilon_1(\theta_0)\epsilon_1(\theta_0)^T},$$
since if the determinant
of two positive definite matrices are the same and they satisfy \eqref{eqn:MSE}, 
then they are equal, \cite{14f6a003-64aa-3900-9a30-7d5cbc9eab0a}. 

What can be said about $\theta$?
It holds for any $\theta$ that 
\begin{align*}
&\E{\epsilon_1(\theta)\epsilon_1(\theta)^T}-
\E{\epsilon_1(\theta_0)\epsilon_1(\theta_0)^T}\\
&\quad\quad\quad{\color{black}=\E{\left(\epsilon_1(\theta)-\epsilon_1(\theta_0)\right)
\left(\epsilon_1(\theta)-\epsilon_1(\theta_0)\right)^T}.}
\end{align*}
\textcolor{black}{The cross terms vanish since} $\epsilon_1(\theta)-\epsilon_1(\theta_0)=\hat x_{o,1}(\theta_0)-
\hat x_{o,1}(\theta)$,
where $\hat x_{o,1}(\theta_0)$ and $\hat x_{o,1}(\theta)$ are \textcolor{black}{orthogonal to}
$\epsilon_1(\theta_0)$. Hence by Lemma~\ref{lem:informative} we have 
for $\theta$ that minimizes $L(\theta)$ that $\mathbf W\left(e^{i\omega};\theta\right)=\mathbf W\left(e^{i\omega};\theta_0\right)$ for almost  
all $\omega$. By Lemma~\ref{lem:predtoinnovation}, equality of the predictor transfer
matrices implies $\mathbf G_c(z;\theta)=\mathbf G_c(z;\theta_0)$. Assumption~A5
then implies $\mathbf G(z;\theta)=\mathbf G(z;\theta_0)$, which yields the
following theorem.
\begin{theorem}
Assume that assumptions A1--A5 hold and that 
$\mathbf C^i(z)\neq0$ for  $|z|\geq 1$ and for all $i$ and uniformly in $\theta\in\Theta$.
Then
$$(\hat a_N,\hat b_N)\rightarrow (a_0,b_0),\quad N\rightarrow\infty\; 
\mathrm{w.p.\,\, 1},$$
i.e. $(a,b)$ are estimated consistently \textcolor{black}{in
\eqref{opt:ML-timevaying}}.\footnote{
In practice we only need to assume that $\mathbf C^i(z)\neq0$ for  $|z|=1$, uniformly in $\theta\in\Theta$. What will happen is that we will 
identify $\mathbf C$ polynomials that have zeros mirrored in the unit circle if the polynomials that
generated the data had zeros outside the unit circle. Another way to phrase it is that we cannot see
in the observed data if a zero is inside or outside the unit circle, and hence we may just as well
assume it to be inside.}
\end{theorem}
\begin{remark}
At regular points of the observed-data model, it follows from Theorem 9.1 and
(9.31) in \cite{lju87} that $\sqrt{N}\left(
(\hat a_N,\hat b_N)-(a_0,b_0)\right)$ is asymptotically normal and efficient,
provided that the plant parameters are also differentially identifiable.\footnote{
Under the standing assumptions, the observed-data model has a unique rational
innovation representation. At a regular point, let $\nu$ denote local
nonredundant coordinates for its innovation transfer matrix and covariance,
and set $\phi=(a,b,\nu)$. We additionally assume that the map
$(a,b)\mapsto\mathbf G_c(a,b)$ is differentially identifiable at the true
parameter, i.e., $D_{(a,b)}\mathbf G_c(a_0,b_0)[d]=0$ implies $d=0$.
Assumption A2 and $\Sigma_e\succ0$ imply that the innovation covariance is
positive definite; together with A4, these conditions make the Fisher
information for $\phi$ nonsingular. The proof of Theorem 9.1 and (9.31) in
\cite{lju87} then applies to $\phi$, with the routine extension to
vector-valued innovations. Selecting the $(a,b)$ subvector gives the
nuisance-adjusted Cramér-Rao bound.}
\end{remark}
\begin{remark}
\textcolor{black}{For consistency, $e_k$ may be non-Gaussian if the disturbances
are independent and identically distributed over time with zero mean and
finite second moments and the other assumptions hold.}
\end{remark}
\begin{remark}\label{rem:innovation-form}
In case we instead of considering $\theta=(a,b,c,\lambda)$ as 
parameters consider 
$\bar \theta= (a,b,K,\Sigma_\epsilon)$ to be the parameters we have that 
$\bar \theta$ sometimes can be estimated consistently. Notice that this does not mean that
we use $\bar\theta$ as optimization variable. 

Since the predictor is estimated consistently we can from the predictions
and observations compute innovations, and from a realization of these
we may estimate their covariance, which is $\Sigma_\epsilon$. This
can always be done consistently. Let us now look at the special cases in Remark~\ref{rem:all-and-one}.
When $T_o=\BBM I&0\EBM$
it holds that $\Sigma_\epsilon=\Sigma_e$, and hence it is possible to estimate $\lambda$ 
consistently. 

By Lemma~\ref{lem:predtoinnovation}, the monic innovation transfer matrix
$\mathbf G_o$ in \eqref{eqn:Go} is estimated consistently. Since its direct term is
$I$, its strictly proper part $\mathbf G_o-I$ is also estimated consistently.
Since we know that $(a,b)$ are estimated consistently it only remains to investigate
if $K$ is also estimated consistently. The Markov parameters of the 
strictly proper transfer matrix $\mathbf G_o-I$ are therefore estimated
consistently; we denote them by $M_k$.
They satisfy $H_oF_c^{k-1}K=M_k$ for $k\geq 1$. Hence we 
obtain the equation
$$\BBM H_o\\H_oF_c\\ \vdots\\H_oF_c^{n-1}\EBM K=\BBM 
M_1\\M_2\\\vdots\\M_n\EBM$$
for $K$. If the observability matrix has full column rank, then we can solve
for $K$, and then $K$ is also estimated consistently. This is however not true in general. 

When the conditions of 
Theorem~\ref{thm:trivial-sol} 
are satisfied it holds that 
$K$ is 
the product of $G_e=C+B\Upsilon$ and an invertible matrix.
Hence if both $b$ and $K$ are
estimated consistently,  then $c$ is also estimated consistently.
\end{remark}
\section{Discussion}\label{sec:discussion}
As already pointed out, the optimization problems for solving the ML
problem are not very tractable, since they in general involve 
{\it complicating constraints} either
in terms of an algebraic Riccati equation or an algebraic Lyapunov equation. 
It is tempting to just remove the algebraic Riccati equation in 
\eqref{opt:ML-stationary} as a constraint. This will result in a 
relaxation of the optimization problem. This can potentially result in 
overfitting, since there are more free parameters to optimize than what
there should be.  There are however cases for when the Riccati equation
is not needed. One case is discussed in Theorem~\ref{thm:trivial-sol-OE}.
The assumptions of that theorem hold for the network in 
Figure~\ref{fig:dyn-network} for the  case when we observe
$x_{o,k}=(u^1_k,u^3_k)$. 
The predictor transfer functions in \eqref{eqn:predictor-transfer-OE}
are given by
\BEQ\label{eqn:pred-example}
\begin{aligned}
\hat{\mathbf{u}}^1&=\mathbf r^1 +\mathbf G^2\mathbf r^2+\mathbf G^3\mathbf u^3\\
\hat{\mathbf{u}}^3&=\mathbf G^1\mathbf u^1+\mathbf r^3.
\end{aligned}
\EEQ
Since $\Sigma=0$, it holds that $\Sigma_\epsilon=R=\diag(\lambda_2+\lambda_3,\lambda_1)$, and the ML problem is 
therefore equivalent to minimizing
$$\frac{1}{2}\sum_{k=1}^N\left(\frac{\left(u^3_{k}-\hat u^3_k\right)^2}{\lambda_1}+
\frac{\left(u^1_k-\hat u^1_k\right)^2}{\lambda_2+\lambda_3}
+ \ln(\lambda_1)+\ln(\lambda_2\!+\!\lambda_3)\!\right) .$$
Since the objective terms depend on disjoint parameters, the structured ML
criterion separates into the two PEM criteria
$$\frac{1}{2}\sum_{k=1}^N\left(u^3_{k}-\hat u^3_k\right)^2,\quad
\frac{1}{2}\sum_{k=1}^N\left(u^1_k-\hat u^1_k\right)^2$$
which agrees with \cite{VANDENHOF20132994,Dankers_predictor_input}.\footnote{The
reason why we can leave out the $\lambda$-dependence is that
we may first minimize with respect to the other variables for
an arbitrary $\lambda$, and then we can go back and minimize 
with respect to $\lambda$.}
Notice that only the sum $\lambda_2+\lambda_3$ can be estimated consistently.

For the two observed signals in \eqref{eqn:pred-example}, the predictor is
affine in each $\mathbf G^i$. If only $\mathbf u^3$ is observed, eliminating
$\mathbf u^1$ introduces the products in \eqref{eq:example-closed-loop}, and
\eqref{eqn:lin-condition} does not hold. This nonlinearity does not prevent a
PEM formulation; it prevents using the separable affine predictors in
\eqref{eqn:pred-example}. The structured ML criterion retains the common
network parametrization.

Let us now consider the simple example of a single feedback loop 
around one single transfer function. We have the equations
\begin{align*}
\mathbf y=\mathbf G\mathbf u+\mathbf e,\quad
\mathbf u=\mathbf y+\mathbf r.
\end{align*}
We observe $(\mathbf y,\mathbf u)$. We notice that there is no noise
present in the second equation. This is an example of what in \cite{WEERTS2018256} is called 
the "rank-reduced noise" case. 
We have a state space
description of the closed loop system which is
\begin{align*}
\xi_{k+1}&=(F+BH)\xi_k+Br_r+(B+C)e_k\\
\BBM y_k\\u_k\EBM&=\BBM 1\\1\EBM(H\xi_k+e_k)+\BBM 0\\1\EBM r_k.
\end{align*}
For this example, Assumption A2 is not fulfilled. It holds that $y_k$ and
$u_k$ can be expressed in one another via the equation $u_k=y_k+r_k$. 
According to Remark~\ref{rem:row-rank}, one of the observations
can be removed. Hence, if we keep the observation of $\mathbf y$, we obtain
the equation
$$\mathbf y=\mathbf G(\mathbf y+\mathbf r)+\mathbf e.$$
Since we always have that the innovation covariance $\Sigma_\epsilon$
is positive definite by Lemma~\ref{lem:DARE} under Assumption A2,
Remark~\ref{rem:row-rank} permits removal of a redundant observation.

\textcolor{black}{In \eqref{opt:ML-timevaying}, the algebraic Lyapunov
equation \eqref{eqn:Lyapunov} is a complicating constraint because its
solution determines the initial covariance $\Sigma_1$ for each candidate
$\theta$. Once $\Sigma_1$ is specified, the Riccati recursion
\eqref{eq:Rrec} can be evaluated directly. As an approximation, we take
the initial covariance to be zero instead
of solving \eqref{eqn:Lyapunov}. We use this
approximation in the next section, and we will see in the numerical
results that it is a good approximation.}
\section{Equivalent Optimization Problem}\label{sec:equiv}
\textcolor{black}{We now derive an equivalent optimization problem
first presented in \cite{10590797}.} To simplify notation,
we define
$u^i=(u_1^i,\ldots,u_N^i)$, $y^i=(y_1^i,\ldots,y_N^i)$, and $e^i=(e_1^i,\ldots,e_N^i)$.
Moreover, we define  $y=(y^1,\ldots,y^M)$, $u=(u^1,u^2,\ldots,u^M)$, 
$e=(e^1,e^2,\ldots,e^M)$, $r=(r^1,r^2,\ldots,r^m)$, and lower-triangular Toeplitz matrices $T_{a^i}\in\reals^{N\times N}$, $T_{b^i}\in\reals^{N\times N}$ and $T_{c^i}\in\reals^{N\times N}$ whose first columns are
$$ \BBM 1\\a^i\\0\EBM,\quad\BBM 0 \\ b^i\\0\EBM,\quad\BBM 1\\c^i\\0\EBM,$$
respectively, c.f., \cite{Wallin02112014}.
\textcolor{black}{We set $y_k^i=u_k^i=e_k^i=0$ for $k\leq 0$.}
This allows us to express the ARMAX models in \eqref{eqn:ARMAX-model} as
\begin{align}
    \label{eqn:ARMAX-T}
    T_{a^i} y^i=T_{b^i}u^i + T_{c^i} e^i.
\end{align}
We also define two block-diagonal matrices,
\begin{align*} 
    T_y &=  \blkdiag (T_{y^1},\ldots,T_{y^M}),\\
    T_u &= \blkdiag (T_{u^1},\ldots,T_{u^M}),
\end{align*}
where $T_{y^i}=T_{c^i}^{-1}T_{a^i}$ and $T_{u^i}=T_{c^i}^{-1}T_{b^i}$.
This allows us to express the dynamic network model as an instance of 
\begin{align}\label{non_sing_gauss_distrib_DN}
    Ax + b = \begin{bmatrix}
        e \\ 0
    \end{bmatrix},
\end{align}
with
\begin{align} 
A&=\BBM A_1\\A_2\EBM = \BBM T_y & -T_u \\ -\Upsilon\otimes I& I\EBM(P\otimes I) \label{eq:A_matrix_def}\\
x&=(P\otimes I)^T\BBM y\\u\EBM ,\quad
b=\BBM b_1\\b_2\EBM= -\BBM 0\\ \Omega\otimes I\EBM r \notag,
\end{align}
where $A_{1}, A_{2} \in \reals^{MN\times 2MN}$, and where 
$P \in \reals^{2M \times 2M}$ is a permutation matrix that is defined such that the observed parts of $y$ and $u$ correspond to the leading entries of $x$. The vector $b \in \reals^{2MN}$ is partitioned conformable to $A$. Notice that the pdf for $x$ is singular.

The matrix $A$ is square, and using the fact that the matrices $T_{a^i}$ and $T_{c^i}$ are non-singular for all $i$, we see that $T_y$ is non-singular. Thus, $A$ is non-singular if the Schur complement $T_y - T_u(\Upsilon\otimes I)$ is full rank. This is true if the closed loop system is well-posed. 
Notice that $A_1$ depends on the unknown model parameters $\theta$,
but $A_2$ and $b$ do not depend on $\theta$. We will see that this allows us to rewrite the pdf for $x$ 
as a nonsingular pdf and an algebraic constraint. 

Consider \eqref{non_sing_gauss_distrib_DN}, where $x$ is partitioned as 
$x = (x_o,x_m)\in  \reals^{n_o + n_m}$,
corresponding to observed and missing data, respectively, $A(\theta) \in \reals^{2MN \times (n_o + n_m)}$ and $b \in \reals^{2MN}$ are partitioned conformably as
\begin{align*}
    A(\theta) = \begin{bmatrix}
    A_1(\theta) \\ A_2
\end{bmatrix} = \begin{bmatrix}
    A_{1o}(\theta)  & A_{1m}(\theta)\\  A_{2o}  & A_{2m}
\end{bmatrix}, \
b = \begin{bmatrix}
    b_1 \\  b_2
\end{bmatrix}.
\end{align*}
Given a Singular Value Decomposition (SVD) of $A_{2m}$, i.e., 
\begin{align*}
    A_{2m} = \begin{bmatrix}
        U_1 & U_2
    \end{bmatrix}\begin{bmatrix}
    \Sigma_1 & 0 \\ 0 & 0    
    \end{bmatrix}\begin{bmatrix}
        V_1 & V_2
    \end{bmatrix}^T,
\end{align*}
we can rewrite (\ref{non_sing_gauss_distrib_DN}) as
\begin{align*}
    \begin{bmatrix}
        A_{1o} & \Bar{A}_{1m1} & \Bar{A}_{1m2}\\ \Bar{A}_{2o1} & \Sigma_1 & 0 \\ \Bar{A}_{2o2} & 0 & 0
    \end{bmatrix}\begin{bmatrix}
        x_o \\ \Bar{x}_{m1} \\ \Bar{x}_{m2} 
    \end{bmatrix} + \begin{bmatrix}
        b_1 \\  \Bar{b}_{21} \\ \Bar{b}_{22} 
    \end{bmatrix} = \begin{bmatrix}
        e \\ 0 \\0
    \end{bmatrix}
\end{align*}
where
\begin{align}
\label{A2o2bar}
    \Bar{A}_{2o} &= \begin{bmatrix}
        \Bar{A}_{2o1} \\ \Bar{A}_{2o2} 
    \end{bmatrix} = \begin{bmatrix}
        U_1^T \\ U_2^T
    \end{bmatrix}A_{2o}, \\ 
    \label{A1m2bar}
    \Bar{A}_{1m} &= \begin{bmatrix}
        \Bar{A}_{1m1} & \Bar{A}_{1m2} 
    \end{bmatrix} = A_{1m}\begin{bmatrix}
        V_1 & V_2
    \end{bmatrix}
\end{align}
and
\begin{align*}
    \Bar{x}_{m} &= \begin{bmatrix}
        \Bar{x}_{m1} \\ \Bar{x}_{m2} 
    \end{bmatrix} = \begin{bmatrix}
        V_1^T \\ V_2^T
    \end{bmatrix}x_{m}, &
    \Bar{b}_{2} &= \begin{bmatrix}
        \Bar{b}_{21} \\ \Bar{b}_{22} 
    \end{bmatrix} = \begin{bmatrix}
        U_1^T \\ U_2^T
    \end{bmatrix}b_{2}.
\end{align*}
Using  $\Sigma_1$ as a pivot, we can rewrite the system as
\begin{align}
\label{eqn:manipuation1}
    \begin{bmatrix}
        \Bar{A}_{1o} & 0 & \Bar{A}_{1m2}\\ \Bar{A}_{2o1} & \Sigma_1 & 0 \\ \Bar{A}_{2o2} & 0 & 0
    \end{bmatrix}\begin{bmatrix}
        x_o \\ \Bar{x}_{m1} \\ \Bar{x}_{m2} 
    \end{bmatrix} + \begin{bmatrix}
        \Bar{b}_1 \\  \Bar{b}_{21} \\ \Bar{b}_{22} 
    \end{bmatrix} = \begin{bmatrix}
        e \\ 0 \\0
    \end{bmatrix},
\end{align}
where $\Bar{A}_{1o} = A_{1o} - \Bar{A}_{1m1} \Sigma_1^{-1}\Bar{A}_{2o1}$ and $\Bar{b}_1 = b_1 - \Bar{A}_{1m1} \Sigma_1^{-1}\Bar{b}_{21}$. The fact that $A$ is invertible implies that $\Bar{A}_{2o2}$ has full row-rank, and hence there exists an orthogonal matrix $W = \begin{bmatrix} W_1 & W_2 \end{bmatrix}$ such that $\Bar{A}_{2o2}W = \begin{bmatrix}
    \Tilde{A}_{2o21} & 0
\end{bmatrix}$ with $\Tilde{A}_{2o21} = \Bar{A}_{2o2}W_1$ nonsingular. The matrix $W$ can be obtained by means of an SVD or an LQ decomposition of $\Bar{A}_{2o2}$. We then define
\begin{align}
\label{A1o2tilde}
    \begin{bmatrix}
        \Tilde{A}_{1o} \\
        \Tilde{A}_{2o1}
    \end{bmatrix} &= \begin{bmatrix}
        \Tilde{A}_{1o1} & \Tilde{A}_{1o2} \\
        \Tilde{A}_{2o11} & \Tilde{A}_{2o12}
    \end{bmatrix} = \begin{bmatrix}
        \Bar{A}_{1o} \\
        \Bar{A}_{2o1}
    \end{bmatrix}W, \\ \notag
    \Bar{x}_o &= \begin{bmatrix}
        \Bar{x}_{o1} \\
        \Bar{x}_{o2}
    \end{bmatrix} = W^Tx_o,
\end{align}
which allows us to rewrite  \eqref{eqn:manipuation1} as
\begin{align*}
    \begin{bmatrix}
        \Tilde{A}_{1o1} & \Tilde{A}_{1o2}  & 0 & \Bar{A}_{1m2}\\ 
        \Tilde{A}_{2o11} & \Tilde{A}_{2o12} & \Sigma_1 & 0 \\
        \Tilde{A}_{2o21} & 0 & 0 & 0
    \end{bmatrix}\begin{bmatrix}
        \Bar{x}_{o1} \\
        \Bar{x}_{o2} \\
        \Bar{x}_{m1} \\ 
        \Bar{x}_{m2} 
    \end{bmatrix} + \begin{bmatrix}
        \Bar{b}_1 \\  \Bar{b}_{21} \\ \Bar{b}_{22} 
    \end{bmatrix} = \begin{bmatrix}
        e \\ 0 \\0
    \end{bmatrix}.
\end{align*}

\vspace*{-7px}
\noindent Using $\Tilde{A}_{2o21}$ as a pivot, we obtain the system
\begin{align*}
    \begin{bmatrix}
        0 & \Tilde{A}_{1o2}  & 0 & \Bar{A}_{1m2}\\ 
        \Tilde{A}_{2o11} & \Tilde{A}_{2o12} & \Sigma_1 & 0 \\
        \Tilde{A}_{2o21} & 0 & 0 & 0
    \end{bmatrix}\begin{bmatrix}
        \Bar{x}_{o1} \\
        \Bar{x}_{o2} \\
        \Bar{x}_{m1} \\ 
        \Bar{x}_{m2} 
    \end{bmatrix} + \begin{bmatrix}
        \Tilde{b}_1 \\  \Bar{b}_{21} \\ \Bar{b}_{22} 
    \end{bmatrix} = \begin{bmatrix}
        e \\ 0 \\0
    \end{bmatrix},
\end{align*}
where $\Tilde{b}_1 = \Bar{b}_1 - \Tilde{A}_{1o1}\Tilde{A}_{2o21}^{-1} \Bar{b}_{22}$. Finally, we rewrite this system as two systems,
\begin{align}
\label{equiv_non_sing}
    \begin{bmatrix}
        \Tilde{A}_{1o2} & \Bar{A}_{1m2}
    \end{bmatrix} \begin{bmatrix}
        \Bar{x}_{o2} \\
        \Bar{x}_{m2} 
    \end{bmatrix} + \Tilde{b}_1 = e, \\
    \begin{bmatrix}
        \Tilde{A}_{2o11} & \Sigma_1\\
        \Tilde{A}_{2o21} & 0
    \end{bmatrix}\begin{bmatrix}
        \Bar{x}_{o1} \\
        \Bar{x}_{m1}
    \end{bmatrix}+\begin{bmatrix}
        \Bar{b}_{21} + \Tilde{A}_{2o12}\Bar{x}_{o2}\\
        \Bar{b}_{22}
    \end{bmatrix} = \begin{bmatrix}
        0 \\ 0
    \end{bmatrix}.
\end{align}
Only $\bar x_{o2}$ and $\bar x_{m2}$ enter the stochastic equation in
\eqref{equiv_non_sing}. Define
\begin{align*}
z_o&=\bar x_{o2}, & z_m&=\bar x_{m2}, &
\gamma&=\widetilde b_1,\\
\Phi_o&=\widetilde A_{1o2}, &
\Phi_m&=\bar A_{1m2},
\end{align*}
and $\Phi=[\Phi_o\;\Phi_m]$. Since $A$ is nonsingular, $\Phi$ is square
and nonsingular. \textcolor{black}{The covariance of $e$ is
$\Lambda=\diag(\lambda^1I_N,\ldots,\lambda^MI_N)$. Let
$Z=\Phi_m^T\Lambda^{-1}\Phi_m\succ0$ and $v=\Phi_oz_o+\Phi_mz_m+\gamma$.
Applying \cite[Eq.~(11)]{HANSSON20121955} to $\Lambda^{-1/2}$ times
\eqref{equiv_non_sing} gives the objective}
{\color{black}
\begin{equation}\label{opt:ML-equivalent}
\frac12v^T\Lambda^{-1}v+\frac12\log\det\Lambda
-\log|\det\Phi|+\frac12\log\det Z.
\end{equation}}
\textcolor{black}{For fixed $\theta$, minimizing \eqref{opt:ML-equivalent}
over $z_m$ as in \cite[Section~4]{HANSSON20121955} gives the problem we
use for numerical optimization}
{\color{black}
\begin{equation}\label{eq:ML-equivalent-marginal}
\begin{aligned}
\minimize\quad&\frac12q^T\Pi q+\frac12\log\det\Lambda\\
&-\log|\det\Phi|+\frac12\log\det Z\\
\st\quad&\theta\in\Theta,
\end{aligned}
\end{equation}}
where $q=\Phi_oz_o+\gamma$ and
\[
\Pi=\Lambda^{-1}-\Lambda^{-1}\Phi_m
Z^{-1}\Phi_m^T\Lambda^{-1}.
\]
The orthogonal matrices used above depend only on $A_2$
\textcolor{black}{in \eqref{eq:A_matrix_def}} and the measurement
selection, not on $\theta$. Thus their Jacobians are constant, whereas
$-\log|\det\Phi|$ accounts for the change of variables from
$(z_o,z_m)$ to $e$ \textcolor{black}{in \eqref{equiv_non_sing}.}
{\color{black}
\begin{proposition}\label{prop:zero-initial-value-likelihood}
Under A2 and given $r$, the problem obtained by replacing the Lyapunov
initialization \eqref{eqn:Lyapunov} in \eqref{opt:ML-timevaying} with zero
initial covariance has the same minimizers in $\theta$ as
\eqref{eq:ML-equivalent-marginal}.
\end{proposition}
\begin{IEEEproof}
Zero initial values make the initial state $\xi_1$ known to be zero, so
its mean $\check\xi_1$ and covariance $\Sigma_1$ are zero. The likelihood
derivation in Section~\ref{sec:maximum-likelihood} applies to the Riccati
recursion \eqref{eq:Rrec} with these initial values and gives the objective
in \eqref{opt:ML-timevaying} up to a constant independent of $\theta$.
Eliminating $z_m$ from \eqref{opt:ML-equivalent} gives
\eqref{eq:ML-equivalent-marginal}, which evaluates the negative
log-likelihood of the same observations for the same exogenous signal and
initial values, up to a constant independent of $\theta$. Thus the two
objectives have the same minimizers in $\theta$.
\end{IEEEproof}}

If $\Lambda=\lambda I$, define
\[
\Pi_0=I-\Phi_m(\Phi_m^T\Phi_m)^{-1}\Phi_m^T,
\qquad d_o=\dim z_o.
\]
Then $\Pi=\lambda^{-1}\Pi_0$, and minimization of \textcolor{black}{the objective in}
\eqref{eq:ML-equivalent-marginal} with respect to $\lambda$ gives
\begin{equation}\label{eq:lambda-concentrated}
\hat\lambda(\theta)=\frac{q^T\Pi_0q}{d_o}.
\end{equation}
Thus the common variance can be eliminated by substitution. The first two
stages of the numerical procedure use this concentrated likelihood.
\section{Numerical Results}
Our implementation is based on the Python library JAX \cite{jax}, which makes use of automatic differentiation to compute the partial derivatives of the cost function. 
The optimization problem associated with the ML approach is non-convex. To ensure convergence to a useful solution, an appropriate initialization strategy is therefore required.
We first solve an ML problem where we consider \textcolor{black}{$c^i=0$}, i.e. 
we approximate
each system with an ARX model and set $\Lambda=\lambda I$. The common variance
is eliminated using \eqref{eq:lambda-concentrated}. This first-stage problem
remains nonconvex because the network equations couple the coefficients and
unrecorded signals. We then solve an ARMAX
problem with the same common-variance elimination and
we initialize the solver with the value of $(a,b)$ we obtained in the previous
step. This will provide us with initial values $(a,b,c,\lambda)$ for a third 
final step, where we do not constrain $\lambda^i$. 
We solve all the optimization problems with a trust-region method implemented via SciPy’s optimize module, see \cite{SciPy}. As in standard unconstrained
estimation of a scalar ARMAX model, the implementation does not use a
stability-preserving parametrization. The staged initialization reduces this
risk but does not guarantee a stable solution.

We will now illustrate some properties of the proposed method using a numerical example based on the network in Fig.~\ref{fig:dyn-network}. 
We generate 100 random models $\mathbf M(z;\theta)$ using the Matlab command
{\tt drss}. For each random model there are three sub-models with two inputs, one for $\mathbf u^i$ and
one for $\mathbf e^i$, and one output for $\mathbf y^i$. We modify the
transfer functions $\mathbf H^i(z;\theta)$ such that they have static gain 
equal to one. The order of all transfer functions is two. Models that are unstable are discarded. Similarly, any model
for which the closed loop system $\mathbf M_c(z;\theta)$ has poles 
which are further away from the origin than 0.9 is
also discarded. We generate different values for $\lambda_i$ for each model
that are drawn from a uniform distribution on the interval $[0.1\bar\lambda,
\bar\lambda]$, where $\bar\lambda=0.1$. For each model 
we draw $e_k^i$ from zero mean Gaussian distributions with variance 
$\lambda_i$ for $1\leq k\leq N$. The value $N$ is different for different
experiments. The reference values $r_k^i$ are equal to $\pm 1$ with 
equal probability. All generated signals are independent of one another and 
in time. The average SNR is approximately given by 
$1/(0.55\bar\lambda)\approx 25{\rm dB}$. We then generate 
observed data using
\begin{align}\label{eq:gen_data}
    {\mathbf{x}}_o=\mathbf G_c(\theta_0)\mathbf r+\mathbf H_c(\theta_0)
\mathbf e,
\end{align}
where we with subindex zero emphasize that this is a model from which 
we collect data. We use zero initial values. We will consider the 
cases when either $(\mathbf u^1,\mathbf u^3)$ or $\mathbf u^3$ are
observed. 

\textcolor{black}{The recovery formulas following \eqref{eq:example-closed-loop}
establish A5 for both measurement configurations when the loop is well
posed and $\mathbf G_0^1$ is not identically zero. Both measure $u^3$, and the formulas
use only the transfer functions from $\mathbf r$ to $\mathbf u^3$. The excitation is stationary and ergodic with independent white
unit-variance channels and is independent of $e$, so A4 holds.}

To validate the estimated models $\hat\theta$, the covariance
\begin{align*}
    \text{Cov}\left[\hat\theta \right] &= \E{(\hat \theta -  \Expect \hat\theta)(\hat \theta -  \Expect\hat \theta)^T},
\end{align*}
and the bias,
\begin{align*}
    \text{bias}\left[\hat\theta\right] =  \E{ \hat\theta - \theta_0},
\end{align*}
of the estimates are computed numerically using averages, by performing 100 estimates considering different realizations of $e$ with the same values for $r$. We also compute the Mean Squared Error (MSE)
given by 
\begin{align*}
    \text{MSE}\left[\hat\theta\right]&=
    \E{\|\hat\theta-\theta_0\|_2^2}\\
    &=\Tr\left(\text{Cov}\left[\hat\theta \right] \right)+
\left\|\text{bias}\left[\hat\theta\right]\right\|_2^2.
\end{align*}
For computing fit, we generate new data according to \eqref{eq:gen_data}, where the noise and reference values are different from the ones  we used
to obtain the model estimate. 
The simulation data for a model $\theta$ is
defined as 
$$\hat{\mathbf{x}}_{o,s}=\mathbf G_c(\theta)\mathbf r.$$
Notice that no noise realization is needed to generate the simulation
data. The same reference value is used as the one used for 
generating the validation data.
The prediction data for the model $\theta$ is defined as 
$$\hat{\mathbf{x}}_{o,p}=\mathbf W_r(\theta)\mathbf r+
\mathbf W_o(\theta){\mathbf{x}}_o.$$
Here there will be noise present, since the latter term contains the 
noise from the validation data. 
The fit is defined as 
$${\rm fit}(\hat x_{o})=1-\frac{\|\hat x_o^i-x_o^i\|_2}
{\|x_o^i-\mathbf{1} \bar x_o^i \|_2},$$
where $\bar x_o^i$ is the mean of $x_o^i$, and where superscript $i$
refers to the $i$th observed output. This definition can be used
for either predicted or simulated data. 

When $(\mathbf u^1,\mathbf u^3)$ is observed, we compare the structured ML
method with a separately fitted PEM based on the following
multiple-input single-output (MISO) equations:
\begin{align}
    \mathbf u^1 - \mathbf r^1 &= \mathbf G^3 \mathbf u^3 + \mathbf G^2\mathbf r^2 + \bar{\mathbf  H} \bar{\mathbf e} \label{eq:dynamic_eq_u1_u3_1} \\
    \mathbf u^3 - \mathbf r^3 &=\mathbf G^1\mathbf u^1+\mathbf H^1\mathbf e^1. \label{eq:dynamic_eq_u1_u3_2}
\end{align}
For the comparison PEM, we use the System Identification Toolbox
\cite{ljung1995system} in MATLAB. The function {\tt armax} estimates
$(\mathbf G^1,\mathbf H^1)$ from \eqref{eq:dynamic_eq_u1_u3_2}, and
{\tt bj} estimates $(\mathbf G^2,\mathbf G^3,\bar{\mathbf H})$ from
\eqref{eq:dynamic_eq_u1_u3_1}. We select the order of $\bar{\mathbf H}$
between one and four using prediction fit on validation data. The comparison
does not estimate $(c,\lambda)$, so we use the two MISO predictors directly:
\begin{align*}
    \hat{\mathbf u}_1 &= \bar{\mathbf  H}^{-1}(\mathbf r^1 + \mathbf G^3 \mathbf u^3 + \mathbf G^2\mathbf r^2) + (1-\bar{\mathbf  H}^{-1})\mathbf u_1 \\ 
    \hat{\mathbf u}_3 &= \mathbf  H_1^{-1}(\mathbf r^3 + \mathbf G^1 \mathbf u^1) + (1-\mathbf  H_1^{-1})\mathbf u_3.
\end{align*}

For both the comparison PEM and structured ML method, the optimization
tolerance is $10^{-5}$.
The MISO comparison is available only with both signals measured. The fits
are similar. Mean times (ML vs.\ PEM) were 15.88 vs.\ 1.18\,s at $N=50$
and 132.06 vs.\ 0.64\,s at $N=500$. The reported ML convergence counts out
of 100, in that order, were 94 and 92 for $u^3$ alone and 100 and 94 for both
signals.
The order of $\bar{\mathbf H}$ for $N=500$ is  1, 2, 3, and 4 for 
31, 39, 20, and 10 of the 100 models, respectively. For $N=50$ the numbers are 
58, 19, 10, and 13, respectively. Lower orders are selected more often for
$N=50$. For the first system, the two
methods give similar fit. The structured ML method has slightly smaller mean
squared error, except that the comparison PEM has smaller covariance for
$N=50$. These results do not establish the cause of the difference in bias.

There is a small degradation in fit with respect
to predicted data when only one signal is measured as compared to the case when 
two signals are measured. There
is also a degradation with respect MSE, covariance and bias of the estimates. 
However, the overall conclusion is that the ML approach is able to successfully obtain a good model
by only measuring $\hat u^3$.\footnote{Our code and data are available on GitHub at: \url{https://github.com/Jvictormata/mle_dyn_net}}
\begin{table*}[htbp]
    \centering
  \hspace*{\dimexpr -\oddsidemargin}
    \caption{Comparison of the PEM and ML for 100 random systems: mean fit values for both simulation and prediction data.}
    \label{tab:table_PEM_vs_MLE}
   \setlength\tabcolsep{2pt}
   \scalebox{0.85}{
  \begin{tabular}{|c|c|c|c|c|c|c|c|c|c|c|c|c|c|c|c|c|c|c|c|c|c|c|c|}
    \hline
     \multirow{3}{*}{N} & \multirow{3}{*}{\shortstack{Observed\\ signals}} &
      \multicolumn{6}{|c|}{PEM} &
      \multicolumn{6}{|c|}{ML} \\
    & & \multicolumn{4}{|c|}{$100\times (\text{avg. Fit}\pm\text{std})$} & \multirow{2}{*}{conv.} & \multirow{2}{*}{\shortstack{avg.\\ time (s)}} & \multicolumn{4}{|c|}{$100\times(\text{avg. Fit}\pm\text{std})$} & \multirow{2}{*}{conv.} & \multirow{2}{*}{\shortstack{avg.\\ time (s)}}\\
    & & $\hat u_{1_\text{pred}}$ & $\hat u_{1_\text{sim}}$ & $\hat u_{3_\text{pred}}$ & $\hat u_{3_\text{sim}}$ & & & $\hat u_{1_\text{pred}}$ & $\hat u_{1_\text{sim}}$ & $\hat u_{3_\text{pred}}$ & $\hat u_{3_\text{sim}}$ & & \\
    \hline
     \multirow{1}{*}{50} & \multirow{2}{*}{$u^3$} & \multirow{1}{*}{--} & \multirow{1}{*}{--}  & \multirow{1}{*}{--} & \multirow{1}{*}{--} & \multirow{1}{*}{--} & \multirow{1}{*}{--} & \multirow{1}{*}{--} & \multirow{1}{*}{--}  & 77.34 $\pm$ 13.05 & 69.32 $\pm$ 18.04 & 94 & 15.75\\
     \multirow{1}{*}{500} &  & \multirow{1}{*}{--} & \multirow{1}{*}{--}  & \multirow{1}{*}{--} & \multirow{1}{*}{--} & \multirow{1}{*}{--} & \multirow{1}{*}{--} & \multirow{1}{*}{--} & \multirow{1}{*}{--} & 81.90 $\pm$ 13.92 & 75.13 $\pm$ 15.45 & 92 & 122.83 \\
     \hline
    \multirow{1}{*}{50} & \multirow{2}{*}{$u^1$ and $u^3$} & \multirow{1}{*}{84.36 $\pm$ 11.46} & \multirow{1}{*}{74.58 $\pm$ 14.97} & \multirow{1}{*}{84.86 $\pm$ 9.96} & \multirow{1}{*}{69.49 $\pm$ 42.48} & \multirow{1}{*}{100} & \multirow{1}{*}{1.18} & 84.89 $\pm$ 10.19 & 74.88 $\pm$ 15.26 & 86.21 $\pm$ 8.48 & 70.64 $\pm$ 33.60 & 100 & 15.88 \\ 
    \multirow{1}{*}{500} &   & \multirow{1}{*}{87.57 $\pm$ 8.12} & \multirow{1}{*}{76.56 $\pm$ 15.52} & \multirow{1}{*}{87.85 $\pm$ 7.52} & \multirow{1}{*}{74.11 $\pm$ 17.35} & \multirow{1}{*}{100} & \multirow{1}{*}{0.64} & 87.09 $\pm$ 8.15 & 76.92 $\pm$ 15.40 & 87.94 $\pm$ 7.40 & 74.91 $\pm$ 16.46 & 94 & 132.06\\
    \hline
  \end{tabular}}
\end{table*}
\begin{table*}[htbp]
    \centering
  \hspace*{\dimexpr -\oddsidemargin}
    \caption{Comparison of the PEM and ML for system 1: fit values for both simulation and prediction data, and trace and max. eigenvalue of the covariance matrix.}
    \label{tab:table_PEM_vs_MLE_system_1}
   \setlength\tabcolsep{3pt}
   \scalebox{0.85}{
  \begin{tabular}{|c|c|c|c|c|c|c|c|c|c|c|c|c|c|c|c|c|c|c|c|c|c|c|c|c|c|c|c|c|c|c|c|c|}
    \hline
     \multirow{3}{*}{N} & \multirow{3}{*}{\shortstack{Observed\\ signals}} &
      \multicolumn{8}{|c|}{PEM} &
      \multicolumn{8}{|c|}{ML} \\
    & & \multicolumn{4}{|c|}{$100\times \text{Fit}$} & \multicolumn{2}{|c|}{Cov. matrix $(a,b)$} & \multirow{2}{*}{$\|\text{bias}(a,b)\|_2^2$} &\multirow{2}{*}{MSE$(a,b)$} & \multicolumn{4}{|c|}{$100\times \text{Fit}$} & \multicolumn{2}{|c|}{Cov. matrix $(a,b)$} & \multirow{2}{*}{$\|\text{bias}(a,b)\|_2^2$}  & \multirow{2}{*}{MSE$(a,b)$} \\
    & & $\hat u_{1_\text{pred}}$ & $\hat u_{1_\text{sim}}$ & $\hat u_{3_\text{pred}}$ & $\hat u_{3_\text{sim}}$ & Trace & Max.\ eigval. & & & $\hat u_{1_\text{pred}}$ & $\hat u_{1_\text{sim}}$ & $\hat u_{3_\text{pred}}$ & $\hat u_{3_\text{sim}}$ & Trace & Max.\ eigval. & & \\
    \hline
     50 & \multirow{2}{*}{$u^3$} & --  & -- & -- & -- & -- & -- & -- & -- & -- & --  & 85.66 & 81.77 & 7.42e-01  & 4.44e-01 & 2.79e-01 & 1.02 \\
     500 & & --  & -- & -- & -- & -- & -- & -- & -- & -- & --  & 88.05 & 84.15 & 3.77e-01  & 3.59e-01 & 2.46e-01 & 6.23e-01 \\
     \hline
    \multirow{1}{*}{50} & \multirow{2}{*}{$u^1$ and $u^3$}  & \multirow{1}{*}{87.37} & \multirow{1}{*}{82.36} & \multirow{1}{*}{87.14} & \multirow{1}{*}{81.48} & \multirow{1}{*}{5.11e-01} & \multirow{1}{*}{4.65e-01} & \multirow{1}{*}{2.96e-01} & \multirow{1}{*}{8.08e-01} & 87.06 & 82.25 & 87.13 & 81.46 & 5.69e-01 & 5.41e-01 & 1.09e-02 & 5.80e-01 \\
    \multirow{1}{*}{500} &  & \multirow{1}{*}{88.12}  & \multirow{1}{*}{81.00} & \multirow{1}{*}{89.16} & \multirow{1}{*}{84.18} & \multirow{1}{*}{5.84e-01} & \multirow{1}{*}{5.82e-01} & \multirow{1}{*}{2.32e-03} & \multirow{1}{*}{5.87e-01} & 88.03 & 81.00 & 89.16 & 84.18 & 4.14e-01 & 4.13e-01 & 1.75e-02 & 4.32e-01 \\
    \hline
  \end{tabular}}
\end{table*}
\section{Conclusions}
We developed a structured maximum-likelihood method for estimating plant
transfer functions in a known network from \textcolor{black}{partial node
measurements}. The estimates \textcolor{black}{from \eqref{opt:ML-timevaying}}
are consistent under the stated assumptions and asymptotically efficient
under additional local regularity and differential-identifiability
conditions. An equivalent \textcolor{black}{formulation with zero initial
covariance avoids the Riccati recursion} during optimization. The numerical
experiments \textcolor{black}{gave good fits with one and two measured signals}.
For two \textcolor{black}{measured} signals, the separately fitted
prediction-error \textcolor{black}{method} was substantially faster and
gave similar fit. The structured method failed to converge in some trials.
\section{Appendix}
\subsection{Uniform Convergence of Riccati Recursion}\label{sec:riccati-convergence}
\begin{lemma}\label{lem:differentiability}
Under assumptions A2--A3, it holds that the solution $(\Sigma,K,\Sigma_\epsilon)$ to \eqref{eq:ARE} 
is a continuously differentiable function of $\theta\in\Theta_o$. 
\end{lemma}
\begin{IEEEproof}
The results follow from the proof of Theorem~4 in \cite{HANSSON1999245}. 
\end{IEEEproof}
\begin{lemma}\label{lemma:transformation_riccati_eq}
Assume that A2 holds and that 
$\mathbf C^i(z)$ are not equal to zero for $|z|\geq 1$. Then it holds that 
$(\bar F_c, H_o)$ is detectable and $(\bar F_c,\bar Q^{1/2})$ is stabilizable.
\end{lemma}
\begin{IEEEproof}
\textcolor{black}{Since $\theta\in\Theta\subset\Theta_o$, zero observer gain
gives detectability of $(F_c,H_o)$. Detectability is invariant under output
injection, so $(\bar F_c,H_o)$ is detectable.}

Let $\bar G_e= G_e\Sigma_e^{1/2}$ and $\bar J_{eo}=J_{eo}\Sigma^{1/2}_e$. Consider
$$\bar{\mathbf P}(z)=\mathbf P(z)\BBM I& 0\\0&\Sigma^{1/2}_e\EBM=
\BBM -zI+F_c& \bar G_e\\H_o&\bar J_{eo}\EBM.
$$
Since $\mathbf C^i(z)$ are not equal to zero for $|z|\geq 1$, 
it follows similarly as in the proof of Lemma~\ref{lem:DARE} that 
$\bar{\mathbf P}(z)$ has full row rank for $|z|\geq 1$. Let $U$ be an orthonormal matrix
such that 
$$\bar{\mathbf P}(z)U=\BBM -zI+F_c& \tilde G_{e1}& \tilde G_{e2}\\H_o&\tilde J_{eo}&0\EBM$$
where $\tilde J_{eo}$ is invertible. Now we have that 
$$\setlength{\arraycolsep}{2pt}\BBM I&-\tilde G_{e1}\tilde J_{eo}^{-1}\\0&I\EBM \! \bar{\mathbf P}(z)U \! = \!
\BBM -zI \! +F_c-\tilde G_{e1}\tilde J_{eo}^{-1}H_o & 0& \tilde G_{e2}\\H_o&\tilde J_{eo}&0\EBM.
$$
From this it follows by the PBH test for stabilizability that 
$(F_c-\tilde G_{e1}\tilde J_{eo}^{-1}H_o,\tilde G_{e2})$ is stabilizable. We have
$S=\bar G_e\bar J_{eo}^T=\tilde G_{e1}\tilde J_{eo}^T$ and $R=\bar J_{eo}\bar J_{eo}^T=
\tilde J_{eo}\tilde J_{eo}^T$ and hence $\bar F_c=F_c-SR^{-1}H_o=F_c-\tilde G_{e1}\tilde J_{eo}^{-1}H_o$.
Moreover, $Q=\bar G_e\bar G_e^T=\tilde G_{e1}\tilde G_{e1}^T+\tilde G_{e2}\tilde G_{e2}^T$ and
hence it holds that $\bar Q=Q-SR^{-1}S^T=\tilde G_{e2}\tilde G_{e2}^T$. From this it follows that 
$(\bar F_c,\bar Q^{1/2})$ is stabilizable.
\end{IEEEproof}
Multiply
\eqref{eq:ARE-QRS} with $\BBM I&-K\\0&I\EBM$
from the left and its transpose from the right to obtain the equivalent equations
\begin{align*}
\Sigma &= (F_c-KH_o)\Sigma(F_c-KH_o)^T\\& \quad +Q-KS^T-SK^T+KRK^T\\
0&=(F_c-KH_o)\Sigma H_o^T+S-KR\\
\Sigma_\epsilon&=H_o\Sigma H_o^T+R,
\end{align*}
which equivalently may be written as
\BEQ
\begin{aligned}\label{eq:ARE-K}
\Sigma &= (F_c-KH_o)\Sigma(F_c-KH_o)^T+Q-SR^{-1}S^T\\&\quad+\left(K-SR^{-1}\right)R
\left(K-SR^{-1}\right)^{\textcolor{black}{T}}\\
K&=(F_c\Sigma H_o^T+S)(H_o\Sigma H_o^T+R)^{-1}\\
\Sigma_\epsilon&=H_o\Sigma H_o^T+R.
\end{aligned}
\EEQ
Define the functions
\begin{align*}
    \mathcal{R}(\Sigma,K) &= (F_c-KH_o)\Sigma (F_c-KH_o)^T + Q -SR^{-1}S^T \\ & \quad + (K - SR^{-1})R(K - SR^{-1})^T\\
    \mathcal{G}(\Sigma)&= (F_c\Sigma H_o^T + S)(H_o\Sigma H_o^T+R)^{-1}.
\end{align*}
Then, from \eqref{eq:ARE-K}, the  Riccati recursion in \eqref{eq:Rrec} is equivalent to 
\begin{align*}
\Sigma_{k+1}= \mathcal{R}(\Sigma_k,K_k),\quad
K_k= \mathcal{G}(\Sigma_k)
\end{align*}
and the algebraic Riccati equation in \eqref{eq:ARE} is equivalent to 
\begin{align*}
\Sigma = \mathcal{R}(\Sigma,K),\quad 
K = \mathcal{G}(\Sigma).
\end{align*}
We have the following lemmas.
\begin{lemma}\label{lemma:minimum-property}
Let $0\preceq \Sigma_{1} \preceq \bar \Sigma_{1}$, let $\bar\Sigma_{k+1} = \mathcal{R}(\bar \Sigma_k,\bar K_k)$ with $\bar K_k$ arbitrary and let $\Sigma_{k+1} = \mathcal{R}(\Sigma_k, K_k)$ with $K_k = \mathcal{G}(\Sigma_k)$, then
\begin{align*}
    \Sigma_{k+1} \preceq \bar\Sigma_{k+1}, \quad k = 1,2,\dots .
\end{align*}
\end{lemma}
\begin{IEEEproof}
    See Theorem 2.1 in \cite{Caines_minimum_property}.
\end{IEEEproof}
\begin{lemma}\label{lemma:convergence-riccati}
Assume that A2-A3 hold uniformly in $\theta\in\Theta$. Then 
there exist constants $K_\Sigma>0$, finite, and $0<\lambda_\Sigma<1$ such that 
$$\sup_{\theta\in\Theta}\|\Sigma_k(\theta)-\Sigma(\theta)\|_2<K_\Sigma\lambda_\Sigma^k,\quad k\geq \textcolor{black}{1}$$
where $\Sigma_k(\theta)$ is defined by the Riccati recursion in \eqref{eq:Rrec} with
\textcolor{black}{$\Sigma_1(\theta)=P(\theta)$ given by \eqref{eqn:Lyapunov}}, and where $\Sigma(\theta)$
is the solution of the algebraic Riccati equation in \eqref{eq:ARE}. 
\end{lemma}
{\color{black}
\begin{IEEEproof}
By A2 and \eqref{eq:ARE-QRS}, $R\succ0$ and $Q-SR^{-1}S^T\succeq0$.
The first equation in \eqref{eq:ARE-K} gives
$\Sigma\succeq(F_c-KH_o)\Sigma(F_c-KH_o)^T$. Iterating this inequality
and using Lemma~\ref{lem:DARE} gives $\Sigma\succeq0$.
Subtracting the $(1,1)$ block of \eqref{eq:ARE} from \eqref{eqn:Lyapunov}
gives $P-\Sigma=F_c(P-\Sigma)F_c^T+K\Sigma_\epsilon K^T$.
Since $\theta\in\Theta\subset\Theta_o$, $F_c^k\to0$ as $k\to\infty$.
Iteration gives $P\succeq\Sigma$.
Let $\bar\Sigma_1=P$ and $\bar\Sigma_{k+1}=\mathcal R(\bar\Sigma_k,K)$.
The Riccati recursion initialized at $\Sigma$ is constant by
\eqref{eq:ARE}. Lemma~\ref{lemma:minimum-property} compares the recursion
initialized at $P$ with this constant recursion and with
$\bar\Sigma_k$, giving
\[
0\preceq\Sigma_{k+1}-\Sigma\preceq\bar\Sigma_{k+1}-\Sigma.
\]
Subtracting $\Sigma=\mathcal R(\Sigma,K)$ at each step gives
$\bar\Sigma_{k+1}-\Sigma=(F_c-KH_o)^k(P-\Sigma)((F_c-KH_o)^k)^T$.
By Lemmas~\ref{lem:DARE} and~\ref{lem:differentiability}, each $\theta$
has a neighborhood where $\|(F_c-KH_o)^m\|_2<1/2$ for some $m\geq1$.
Let $m_\star$ be the largest such $m$ in a finite subcover of $\Theta$.
By continuity, the norms of powers below $m_\star$ are bounded on
$\Theta$ by some $C_\Sigma$. Writing $k=jm+s$, $0\le s<m$, gives
$\|(F_c-KH_o)^k\|_2\le C_\Sigma2^{-j}\le 2C_\Sigma 2^{-k/m_\star}$.
The bound on $\Sigma_{k+1}-\Sigma$ proves the claim with
$\lambda_\Sigma=2^{-2/m_\star}$ and
$K_\Sigma>4C_\Sigma^2\sup_{\theta\in\Theta}\|P-\Sigma\|_2/\lambda_\Sigma$.
The supremum is finite by continuity from \eqref{eqn:Lyapunov} and
Lemma~\ref{lem:differentiability}.
\end{IEEEproof}}
\begin{remark}
Notice that the constants are independent of $\theta$, and 
hence this shows uniform geometric convergence.
\end{remark}
\begin{lemma}\label{lem:bound-cov}
Assume that A2-A3 hold uniformly in $\theta\in\Theta$.
Let $P(\theta)$ be the solution of the algebraic Lyapunov equation in \eqref{eqn:Lyapunov}, let 
$\Sigma_k(\theta)$ be the solution of the Riccati recursion in \eqref{eq:Rrec} with $\Sigma_1(\theta)=P(\theta)$, and 
let $\Sigma(\theta)$ be the solution of the algebraic Riccati equation in 
\eqref{eq:ARE}. Then it holds that $\Sigma_k(\theta)\succeq \Sigma(\theta)$ and that $\Sigma_{\epsilon,k}(\theta)\succeq \Sigma_\epsilon(\theta)\succ 0$ 
uniformly for all $k\geq 1$ 
and $\theta\in\Theta$. 
\end{lemma}
\begin{IEEEproof}
\textcolor{black}{The proof of Lemma~\ref{lemma:convergence-riccati} gives
$\Sigma_k\succeq\Sigma$ for all $k\geq1$. By \eqref{eq:Rrec} and
\eqref{eq:ARE}, $\Sigma_{\epsilon,k}-\Sigma_\epsilon=
H_o(\Sigma_k-\Sigma)H_o^T\succeq0$. By Lemmas~\ref{lem:DARE} and
\ref{lem:differentiability}, $\Sigma_\epsilon$ is positive definite
and continuous on compact $\Theta$, so its smallest eigenvalue is
bounded away from zero.}
\end{IEEEproof}
\begin{corollary}\label{cor:ric}
Assume that A2-A3 hold uniformly in $\theta\in\Theta$.
Define $Q_k(\theta) = \Sigma_{\epsilon,k}(\theta)^{-1} - \Sigma_{\epsilon}(\theta)^{-1}$,  $\Delta \Sigma_{\epsilon,k}(\theta) = \Sigma_{\epsilon,k}(\theta)-\Sigma_{\epsilon}(\theta)$ and $\Delta K_k(\theta)=K(\theta)-K_k(\theta)$. As a direct consequence of Lemmas~\ref{lemma:convergence-riccati} and \ref{lem:bound-cov}, it holds:
\begin{enumerate}
    \item $\|Q_k(\theta)\|_2\rightarrow 0$, geometrically as $k\rightarrow\infty$ and uniformly in $\theta\in\Theta$,
    \item $\|\Sigma_{\epsilon}(\theta)^{-1}\Delta\Sigma_{\epsilon,k}(\theta)\|_2\rightarrow0$, geometrically as $k\rightarrow\infty$ and uniformly in $\theta\in\Theta$,
    \item $\|\Delta K_k(\theta)\|_2 \leq M_\Sigma\lambda_\Sigma^k$,
\end{enumerate}
with finite $M_\Sigma>0$, and  $0 < \lambda_\Sigma < 1$.
\end{corollary}

\begin{IEEEproof}
From Lemma~\ref{lemma:convergence-riccati} and recalling that
\begin{align*}
    \Sigma_{\epsilon,k} &= H_o \Sigma_k H_o^T + J_{eo}\Sigma_{e}J_{eo}^T
    \\ \Sigma_{\epsilon} &= H_o\Sigma H_o^T+\; J_{eo}\Sigma_{e} J_{eo}^T,
\end{align*}
we can write 
\begin{align*}
    \|\Sigma_{\epsilon,k} - \Sigma_{\epsilon}\|_2 \leq \|H_o\|^2_2 \|\Sigma_k - \Sigma \|_2 \leq \|H_o\|^2_2 K_{\Sigma} \lambda_{\Sigma}^k.
\end{align*}
Using the identity $A^{-1} - B^{-1} = A^{-1}(B - A)B^{-1}$, we can write 
\begin{align*}
    \|Q_k\|_2 = \|\Sigma_{\epsilon,k}^{-1} - \Sigma_{\epsilon}^{-1}\|_2 \le \|\Sigma_{\epsilon,k}^{-1}\|_2 \, \|\Sigma_{\epsilon}^{-1}\|_2 \, \|\Sigma_{\epsilon,k} - \Sigma_{\epsilon}\|_2.
\end{align*}
Since, from Lemma~\ref{lem:bound-cov}, $\|\Sigma_{\epsilon,k}(\theta)^{-1}\|_2$ and $\|\Sigma_{\epsilon}(\theta)^{-1}\|_2$ are both uniformly bounded for $\theta\in\Theta$ and $k\ge0$, we have
\begin{align*}
   \|Q_k\|_2 \le \Bigl( \sup_{\substack{k,\theta}} \|\Sigma_{\epsilon,k}^{-1}\|_2 \, \|\Sigma_{\epsilon}^{-1}\|_2 \,  \|H_o\|^2_2 \Bigr)\, K_{\Sigma} \lambda_{\Sigma}^k = K_{\epsilon} \lambda_{\Sigma}^k.
\end{align*}
for some finite $K_\epsilon$. 
Then as a direct consequence of Lemma~\ref{lemma:convergence-riccati}, $\|Q_k(\theta)\|_2 \to 0$, geometrically as $k\to\infty$ and 
uniformly in $\theta \in \Theta$. 

We have that 
$$\Sigma_{\epsilon}(\theta)^{-1}\Delta\Sigma_{\epsilon,k}(\theta)=
-\Sigma_{\epsilon}(\theta)^{-1}\Sigma_{\epsilon,k}(\theta)Q_k(\theta)\Sigma_{\epsilon}(\theta)$$
Notice that $\Sigma_{\epsilon}(\theta)$ and $\Sigma_{\epsilon,k}(\theta)$ are
uniformly bounded in $k$ and $\theta\in\Theta$. 
Since, from Lemma~\ref{lem:bound-cov}, $\|\Sigma_{\epsilon}(\theta)^{-1}\|_2$ is uniformly bounded for $\theta\in\Theta$, we also have $\|\Sigma_{\epsilon}(\theta)^{-1}\Delta\Sigma_{\epsilon,k}(\theta)\|_2\rightarrow0$, geometrically as $k\to\infty$ and uniformly in $\theta \in \Theta$.

From $\Delta K_k(\theta)=K(\theta)-K_k(\theta)$, we have
\begin{align*}
    \Delta K_k(\theta) &= \bigl(F_c\Sigma H_o^T + S\bigr) \bigl(H_o\Sigma H_o^T + R\bigr)^{-1} \\ &\quad - \bigl(F_c\Sigma_k H_o^T + S\bigr) \bigl(H_o\Sigma_k H_o^T + R\bigr)^{-1}.
\end{align*}
A straightforward rearrangement yields
\begin{align*}
    &\Delta K_k = (F_c(\Sigma-\Sigma_k)H_o^T)\bigl(H_o\,\Sigma\,H_o^T + R\bigr)^{-1} \\
  &+ \! \bigl(F_c\Sigma_k H_o^T \! + S\bigr)\! \Bigl[ \bigl(H_o\Sigma H_o^T + R\bigr)^{-1}\!\! - \bigl(H_o\Sigma_k H_o^T + R\bigr)^{-1} \Bigr].
\end{align*}
Using again the identity $A^{-1} - B^{-1} = A^{-1}(B - A)B^{-1}$, we get 
\begin{align*}
    &\bigl(H_o\Sigma H_o^T + R\bigr)^{-1} - \bigl(H_o\Sigma_k H_o^T + R\bigr)^{-1} = \\
    &\quad\quad \bigl(H_o\Sigma H_o^T + R\bigr)^{-1} H_o\bigl(\Sigma_k -\Sigma \bigr)H_o^T\bigl(H_o\Sigma_k H_o^T + R\bigr)^{-1}.
\end{align*}
Using the triangular inequality, we have
\begin{align*}
    &\|\Delta K_k(\theta)\|_2 \leq \|F_c\|_2\|\Sigma -\Sigma_k\|_2\|H_o\|_2\|(H_o\Sigma H_o^T + R)^{-1}\|_2 \\ 
    &+ \|F_c\Sigma_k H_o^T + S\|_2 \|(H_o\Sigma H_o^T + R)^{-1}\|_2 \|H_o\|_2\|\Sigma -\Sigma_k \|_2 \\
    &\times \|H_o\|_2\|\bigl(H_o\Sigma_k H_o^T + R\bigr)^{-1}\|_2.
\end{align*}
\textcolor{black}{By Lemmas~\ref{lemma:convergence-riccati} and \ref{lem:bound-cov},
$\Sigma_k$, $\Sigma$, $\Sigma_{\epsilon,k}^{-1}$, and
$\Sigma_\epsilon^{-1}$ are uniformly bounded. Since $F_c$ and $S$ are
continuous on compact $\Theta$, there are finite constants $M>0$ and
$M_\Sigma>0$ such that}
\begin{align*}
    \|\Delta K_k(\theta)\|_2\leq M \|\Sigma(\theta)-\Sigma_k(\theta)\|_2 \leq M_\Sigma\lambda_\Sigma^k.
\end{align*}
Then, since $\|\Sigma_k(\theta) - \Sigma(\theta)\|$ decays with rate $\lambda_{\Sigma}$, where  $0<\lambda_{\Sigma}<1$, we conclude $\|K(\theta) - K_k(\theta)\|$ also decays with the same rate $\lambda_{\Sigma}$.
\end{IEEEproof}
\begin{lemma}\label{lem:stabilizing_sequence}
    Assume that A2--A3 hold uniformly in $\theta\in\Theta$. Initialize the Riccati recursion in \eqref{eq:Rrec} with 
    $\Sigma_1(\theta)=P(\theta)$, where $P(\theta)$ is the solution of the algebraic Lyapunov equation 
    in \eqref{eqn:Lyapunov}. Then it holds for the sequence of 
    $\Sigma_k(\theta)$ defined by the Riccati recursion that  
    $$\max_i|\lambda_i(F_c(\theta) - K_k(\theta) H_o)|<\lambda_{\max}<1$$
    uniformly in  $k\geq 1$ and $\theta\in\Theta$.
\end{lemma}
{\color{black}
\begin{IEEEproof}
By \eqref{eq:Rrec} and \eqref{eqn:Lyapunov}, $P-\Sigma_2
=K_1\Sigma_{\epsilon,1}K_1^T\succeq0$ with
$K_1=(F_cPH_o^T+S)\Sigma_{\epsilon,1}^{-1}$.
Apply Lemma~\ref{lemma:minimum-property} to the Riccati sequences starting
at $\Sigma_2$ and $P$ to get $\Sigma_{k+1}\preceq\Sigma_k$ for $k\geq1$.
Subtracting $\mathcal R(\Sigma,K)$ from $\mathcal R(\Sigma_k,K_k)$
and using the second equation in \eqref{eq:ARE-K} to cancel the terms
linear in $\Delta K_k$ gives
\begin{align*}
\Sigma_{k+1}-\Sigma
&=(F_c-K_kH_o)(\Sigma_k-\Sigma)(F_c-K_kH_o)^T\\
&\quad +\Delta K_k\Sigma_\epsilon\Delta K_k^T.
\end{align*}
If $\mu$ is an eigenvalue of $F_c-K_kH_o$ with $|\mu|\geq1$, take a
corresponding left eigenvector $v^*$. Multiply the identity by $v^*$
on the left and $v$ on the right. Using $\Sigma_{k+1}\preceq\Sigma_k$
and $\Sigma_k\succeq\Sigma$ from Lemma~\ref{lem:bound-cov} gives
\[
0\le v^*\Delta K_k\Sigma_\epsilon\Delta K_k^Tv
\le (1-|\mu|^2)v^*(\Sigma_k-\Sigma)v\le0.
\]
Since $\Sigma_\epsilon\succ0$, $v^*\Delta K_k=0$, so
$v^*(F_c-KH_o)=\mu v^*$ contradicts Lemma~\ref{lem:DARE}.
By Lemma~\ref{lem:DARE} and compactness of $\Theta$, the eigenvalues
of $F_c-KH_o$ have moduli uniformly bounded below one. The maximum
modulus is continuous in the matrix, so uniform convergence of $K_k$
by Corollary~\ref{cor:ric} gives a bound below one for large $k$. For
the finitely many remaining $k$, $K_k$ is continuous by \eqref{eq:Rrec}
and \eqref{eqn:Lyapunov}, so compactness gives the uniform bound.
\end{IEEEproof}}
\subsection{Convergence of Quadratic Form}
\textcolor{black}{Let $\epsilon_k^{(0)}(\theta)$ denote the prediction errors
in \eqref{opt:ML-stationary} with $\hat\xi_1=0$ and the stabilizing solution
of \eqref{eq:ARE}.}
\begin{lemma}\label{lem:conv-quadrati-form}
\textcolor{black}{Assume A2--A4.}
With 
$$\mathcal L_N(\theta)=\frac{1}{N}\sum_{k=1}^N\epsilon_k(\theta)^T\Sigma_\epsilon(\theta)^{-1}
\epsilon_k(\theta)$$
it holds w.p. 1 that 
\BEQ\label{eqn:quadratic-form-convergence}
\sup_{\theta\in\Theta}\left|\mathcal L_N(\theta)-
\E{\epsilon_1(\theta)^T\Sigma_\epsilon(\theta)^{-1}
\epsilon_1(\theta)}\right|\rightarrow 0,\quad N\rightarrow \infty.
\EEQ
\textcolor{black}{The same limit holds when the quadratic form uses
$\epsilon_k^{(0)}(\theta)$.}
\end{lemma}
\begin{IEEEproof}
Since $\epsilon_k(\theta)$ is \textcolor{black}{stationary and ergodic as a time-invariant
function of the jointly stationary and ergodic processes $r$ and $e$},
$\Sigma_\epsilon(\theta)\succ 0$ by Lemma~\ref{lem:DARE}, and since
{\color{black}
\begin{align*}
&\E{\epsilon_1(\theta)^T\Sigma_\epsilon(\theta)^{-1}\epsilon_1(\theta)}\\
&\qquad=\Tr\bigl(\Sigma_\epsilon(\theta)^{-1}\E{\epsilon_1(\theta)\epsilon_1(\theta)^T}\bigr)
<\infty,
\end{align*}
by \eqref{eqn:closed-loop}, \eqref{eqn:time-invariant-predictor},
Lemma~\ref{lem:DARE}, and the finite second moments of $r_k$ and $e_k$,}
it holds by \textcolor{black}{Birkhoff's ergodic theorem
\cite[Theorem 6.2.1]{durrett2019}} that
we have point-wise convergence for every $\theta\in\Theta$.
To prove uniform convergence, \textcolor{black}{we bound the partial
derivatives of $\mathcal L_N$ with probability one, uniformly in
$\theta$ and $N$.} It holds that
\BEQ\label{eq:quadratic-form-derivative}
\begin{aligned}
\frac{\partial \mathcal L_N}{\partial \theta_i}&=
\frac{1}{N}\sum_{k=1}^N\frac{\partial\epsilon_k(\theta)^T}{\partial\theta_i}\Sigma_\epsilon(\theta)^{-1}
\epsilon_k(\theta)\\
&-\frac{1}{N}\sum_{k=1}^N\epsilon_k(\theta)^T
\Sigma_\epsilon(\theta)^{-1}\frac{\partial\Sigma_\epsilon(\theta)}{\partial\theta_i}
\Sigma_\epsilon(\theta)^{-1}
\epsilon_k(\theta)\\
&+\frac{1}{N}\sum_{k=1}^N\epsilon_k(\theta)^T\Sigma_\epsilon(\theta)^{-1}
\frac{\partial\epsilon_k(\theta)}{\partial\theta_i}.
\end{aligned}
\EEQ
From now on we will denote derivative with respect to $\theta_i$ with a prime, i.e.
$\epsilon'_k(\theta)=\frac{\partial\epsilon_k(\theta)}{\partial\theta_i}$, and we will also 
sometimes omit the dependence on $\theta$. From \eqref{eqn:innovation} it follows that 
\begin{align*}
\hat\xi'_{k+1}&=\left(F_c-KH_o\right)\hat\xi'_k+F_c'\hat\xi_k+\textcolor{black}{G_r'r_k}+K'\epsilon_k\\
\epsilon'_k&=-H_o\hat\xi'_k,
\end{align*}
{\color{black}
where $\hat\xi_k$ is defined in \eqref{eqn:time-invariant-predictor}.
Substituting $\epsilon_k=x_{o,k}-H_o\hat\xi_k-J_{ro}r_k$ gives a block
triangular state matrix for $(\hat\xi_k,\hat\xi'_k)$ with both diagonal
blocks equal to $F_c-KH_o$. By Lemmas~\ref{lem:DARE} and
\ref{lem:differentiability}, this matrix is continuous in $\theta$
and its eigenvalues are strictly inside the unit circle. The argument
in Lemma~\ref{lemma:convergence-riccati} bounds its powers geometrically
and uniformly in $\theta$. Thus for some $\kappa>0$ and $0<\rho<1$
independent of $\theta$,
\[
\|\epsilon_k\|_2+\|\epsilon'_k\|_2
\le \kappa\sum_{j=-\infty}^k\rho^{k-j}(\|x_{o,j}\|_2+\|r_j\|_2).
\]
By \eqref{eq:quadratic-form-derivative} and Lemmas~\ref{lem:differentiability}
and \ref{lem:bound-cov}, the absolute value of the derivative is bounded
by a constant times the time average of the square of the sum.
Weighted Cauchy--Schwarz and stationarity give
\begin{align*}
&\E{\Bigl(\sum_{j=-\infty}^k\rho^{k-j}(\|x_{o,j}\|_2+\|r_j\|_2)\Bigr)^2}\\
&\qquad\le \frac{\E{(\|x_{o,1}\|_2+\|r_1\|_2)^2}}{(1-\rho)^2}<\infty.
\end{align*}
The squared sum is stationary and ergodic as a time-invariant function
of $r$ and $e$. Thus Birkhoff's ergodic theorem bounds its time averages
almost surely. Applying the same bound to $\|\epsilon_k(\theta)\|_2^2$
gives, with probability one,
\[
C_\epsilon=\sup_{N\geq1,\;\theta\in\Theta}
\frac{1}{N}\sum_{k=1}^N\|\epsilon_k(\theta)\|_2^2<\infty.
\]
By continuity and compactness, $\Theta$ has a compact
neighborhood in $\Theta_o$ on which A2--A3 hold. The derivative bounds hold
on this neighborhood as well.
The point-wise convergence and the derivative bound give
\eqref{eqn:quadratic-form-convergence} by \cite{1101304}.

Let $\mathcal L_N^{(0)}$ denote $\mathcal L_N$ with $\epsilon_k$ replaced
by $\epsilon_k^{(0)}$. Subtracting the predictor recursions in
\eqref{eqn:time-invariant-predictor} and \eqref{opt:ML-stationary} with
$\hat\xi_1=0$ gives
\[
\delta_k(\theta)=\epsilon_k^{(0)}(\theta)-\epsilon_k(\theta)
=H_o(F_c-KH_o)^{k-1}\hat\xi_1(\theta).
\]
Expanding \eqref{eqn:time-invariant-predictor} over past observations and
using the uniform power bound from Lemma~\ref{lemma:convergence-riccati}
bounds $\sup_{\theta\in\Theta}\|\hat\xi_1(\theta)\|_2$ by a constant times
the weighted sum used for $\epsilon_k$ at $k=1$. The sum has finite second
moment by the weighted Cauchy--Schwarz estimate in this proof. The same
power bound in the formula for $\delta_k$ gives, almost surely,
\[
C_0=\sup_{\theta\in\Theta}\sum_{k=1}^\infty\|\delta_k(\theta)\|_2^2<\infty.
\]
Expanding the quadratic forms and applying Cauchy--Schwarz gives
\begin{align*}
&\sup_{\theta\in\Theta}|\mathcal L_N^{(0)}(\theta)-\mathcal L_N(\theta)|\\
&\qquad\le \sup_{\theta\in\Theta}\|\Sigma_\epsilon(\theta)^{-1}\|_2
\left(2\sqrt{\frac{C_\epsilon C_0}{N}}+\frac{C_0}{N}\right).
\end{align*}
Lemma~\ref{lem:bound-cov} bounds the inverse covariance norm. Since
$C_\epsilon$ and $C_0$ are finite almost surely, the right-hand side tends
to zero. Fubini's theorem gives the same uniform
convergence almost surely with respect to $e$ for almost every realization
of $r$.
}
\end{IEEEproof}

\textcolor{black}{Under A2--A3, the objective in \eqref{opt:ML-stationary} with
 $\hat\xi_1=0$ is continuous in $\theta$. Lemma~\ref{lem:differentiability} gives continuity of the
stabilizing Riccati solution, Lemma~\ref{lem:DARE} gives positive
definiteness of $\Sigma_\epsilon$, and the finite predictor recursion in
\eqref{opt:ML-stationary} gives continuity of each $\epsilon_k^{(0)}$.
A minimum exists on compact $\Theta$.}

\subsection{Convergence of Log-Likelihood Function}
We have the following lemma:
\begin{lemma}\label{lemma:L-converges-uniformly}
Assume that A2-A3 hold uniformly in $\theta\in\Theta$.
\textcolor{black}{Assume A4. In \eqref{eq:Rrec},
let $\Sigma_1=P$ with $P$ given by \eqref{eqn:Lyapunov}.}
It then holds that 
$$\sup_{\theta\in\Theta}\left|L_N(\theta)-L(\theta)\right|\rightarrow 0,\quad 
N\rightarrow \infty\;\mathrm{w.p. 1}$$
where 
\BEQ\label{eqn:L-limit}
L(\theta)=-\E{\ln f(\epsilon_1(\theta);\theta)}
\EEQ
and where
$$
f(\epsilon;\theta)=\frac{1}{\sqrt{(2\pi)^{n_0}\det \Sigma_\epsilon(\theta)}}
\exp\left(-\frac{1}{2}\epsilon^T\Sigma_\epsilon(\theta)^{-1}
\epsilon\right).$$
\end{lemma}
\begin{IEEEproof}
We follow the lines of the proof in \cite{14f6a003-64aa-3900-9a30-7d5cbc9eab0a}. Let 
$$\Delta'_N(\theta)=\frac{1}{N}\sum_{k=1}^N\left(\ln\det \Sigma_{\epsilon,k}(\theta)-\ln\det 
\Sigma_\epsilon(\theta)\right),\quad N\geq 1$$
and 
\begin{align*}
\Delta''_N(\theta)=\frac{1}{N}\sum_{k=1}^N\bigg(&\check\epsilon_k(\theta)^T\Sigma_{\epsilon,k}
(\theta)^{-1}\check\epsilon_k(\theta)
 \\ &-\epsilon_k(\theta)^T\Sigma_{\epsilon}
(\theta)^{-1}\epsilon_k(\theta)\bigg).
\end{align*}
From $\Sigma_{\epsilon,k}(\theta)=\Sigma_{\epsilon}(\theta)+\Delta \Sigma_{\epsilon,k}(\theta)=
\Sigma_{\epsilon}(\theta)\left(I+\Sigma_{\epsilon}(\theta)^{-1}\Delta\Sigma_{\epsilon,k}(\theta)\right)$,
with obvious definition of $\Delta \Sigma_{\epsilon,k}(\theta)$, we have that 
$$\Delta'_N(\theta)=\frac{1}{N}\sum_{k=1}^N\ln\det\left(I+ \Sigma_{\epsilon}(\theta)^{-1}\Delta\Sigma_{\epsilon,k}(\theta)\right).$$
By Corollary~\ref{cor:ric} it holds that 
$\|\Sigma_{\epsilon}(\theta)^{-1}\Delta\Sigma_{\epsilon,k}(\theta)\|\rightarrow0$ both uniformly and 
geometrically as $k\rightarrow\infty$. Hence the eigenvalues of 
$I+ \Sigma_{\epsilon}(\theta)^{-1}\Delta\Sigma_{\epsilon,k}(\theta)$ converges to $1$
uniformly in $\theta\in\Theta$ and geometrically in $k$ as $k\rightarrow \infty$. Hence we have that
$$\sup_{\theta\in\Theta}\left|\Delta'_N(\theta)\right|\rightarrow 0,\quad 
N\rightarrow \infty.$$
We now focus on $\Delta''_N(\theta)$. Let 
\begin{align*}
d_k(\theta)&=\check\epsilon_k(\theta)-\epsilon_k(\theta)\\
Q_k(\theta)&=\Sigma_{\epsilon,k}(\theta)^{-1}-\Sigma_{\epsilon}(\theta)^{-1}.
\end{align*}
We can then write
\BEQ\label{eqn:delta-biss}
\begin{aligned}
&\Delta''_N(\theta)=\frac{1}{N}\sum_{k=1}^N\bigg(
\epsilon_k(\theta)^TQ_k(\theta)\epsilon_k(\theta)\\& \quad\quad +2d_k(\theta)^T
\Sigma_{\epsilon,k}(\theta)^{-1}\epsilon_k(\theta)+d_k(\theta)^T\Sigma_{\epsilon,k}(\theta)^{-1}
d_k(\theta)\bigg).
\end{aligned}
\EEQ
{\color{black}
With $C_\epsilon$ from the proof of Lemma~\ref{lem:conv-quadrati-form},
$\|\epsilon_k(\theta)\|_2^2\leq k C_\epsilon$ almost surely. By
Corollary~\ref{cor:ric}, the first sum in \eqref{eqn:delta-biss}
is bounded in absolute value by
\[
\frac{C_\epsilon}{N}\sum_{k=1}^\infty k\sup_{\theta\in\Theta}
\|Q_k(\theta)\|_2 \longrightarrow 0 .
\]
}

By Lemma~\ref{lem:bound-cov} we have that $\Sigma_{\epsilon,k}(\theta)^{-1}<\kappa_1 I$ for some
positive $\kappa_1$. Hence the absolute value of the second sum in \eqref{eqn:delta-biss}  is by the Cauchy-Schwartz inequality  majorized by
$$2\kappa_1\left(\frac{1}{N}\sum_{k=1}^Nd_k(\theta)^Td_k(\theta)\right)^{1/2}
\left(\frac{1}{N}\sum_{k=1}^N\epsilon_k(\theta)^T\epsilon_k(\theta)\right)^{1/2}$$
and the third sum by 
$$\frac{\kappa_1}{N}\sum_{k=1}^Nd_k(\theta)^Td_k(\theta).$$
Therefore, in order to prove that the second and third sum converge to zero, it is sufficient to 
prove that 
\BEQ\label{eqn:d-squared-limit}
\frac{1}{N}\sum_{k=1}^Nd_k(\theta)^Td_k(\theta)\rightarrow 0,\quad N\rightarrow\infty\;\mathrm{w.p. 1}
\EEQ
and uniformly in $\theta\in\Theta$. From the definitions of $d_k(\theta)$, $\epsilon_k(\theta)$ and
$\check\epsilon_k(\theta)$ it follows with $\tilde\xi_k(\theta)=\hat\xi_k(\theta)-\check\xi_k(\theta)$
that 
\begin{align*}
\tilde\xi_{k+1}(\theta)&=\tilde F_k(\theta)\tilde\xi_k+\Delta K_k(\theta)
\epsilon_k(\theta)\\
d_k(\theta)&=H_o\tilde\xi_k(\theta),
\end{align*}
where $\tilde F_k(\theta)=F_c(\theta)-K_{k}(\theta)H_o$ and $\Delta K_k(\theta)=K(\theta)-K_k(\theta)$.
The initial value is $\tilde\xi_1(\theta)=\hat\xi_1(\theta)$. Define the transition matrix
$$\Phi_{k,j}(\theta)=\tilde F_{k-1}(\theta)\tilde F_{k-2}(\theta)\cdots\tilde F_j(\theta)
\textcolor{black}{,\qquad k>j.}$$
\textcolor{black}{Set $\Phi_{j,j}(\theta)=I$.}
Then we can write the solution to the difference equation above as $d_1(\theta)=H_o\hat\xi_1(\theta)$
and for $k\geq 2$
$$d_{k}(\theta)=H_o\Phi_{k,1}(\theta)\hat\xi_1(\theta)+\sum_{j=1}^{k-1}H_o\Phi_{k,j+1}(\theta)
\Delta K_j(\theta)\epsilon_j(\theta).$$
Let $\alpha_1=H_o\hat\xi_1(\theta)$ and $\beta_1=0$ and for $k\geq 2$
\begin{align*}
\alpha_k&=H_o\Phi_{k,1}(\theta)\hat\xi_1(\theta)\\
\beta_k&=\sum_{j=1}^{k-1}H_o\Phi_{k,j+1}(\theta)
\Delta K_j(\theta)\epsilon_j(\theta).
\end{align*}
We obtain from the Cauchy-Schwartz inequality that 
\begin{align*}
&\frac{1}{N}\sum_{k=1}^Nd_k(\theta)^Td_k(\theta)\leq  \\ &\frac{1}{N}\!\Bigg[\sum_{k=1}^N\alpha_k^T\alpha_k +
2\Bigg(\!\sum_{k=1}^N\alpha_k^T\alpha_k\Bigg)^{\frac{1}{2}}\!\!\Bigg(\!\sum_{k=1}^N\beta_k^T\beta_k\Bigg)^{\frac{1}{2}}\!\!\!\! + \!
\sum_{k=1}^N\beta_k^T\beta_k\Bigg].
\end{align*}
It is now sufficient to show that 
$$\frac{1}{N}\sum_{k=1}^N\alpha_k^T\alpha_k\rightarrow 0,\quad 
\frac{1}{N}\sum_{k=1}^N\beta_k^T\beta_k\rightarrow 0$$
as $N\rightarrow\infty$ w.p. 1 and uniformly in $\theta\in\Theta$.

{\color{black}
By the proof of Lemma~\ref{lemma:convergence-riccati}, choose an integer
$h\geq1$ such that $\sup_{\theta\in\Theta}\|(F_c-KH_o)^h\|_2<1/4$.
Corollary~\ref{cor:ric} gives uniform convergence of $\tilde F_k$ to
$F_c-KH_o$. Since matrix multiplication is continuous on bounded
sets of matrices, $\Phi_{j+h,j}$ converges uniformly to $(F_c-KH_o)^h$
as $j\to\infty$. Thus $\sup_{\theta\in\Theta}\|\Phi_{j+h,j}\|_2\leq1/2$
for all $j\geq j_0$ for some $j_0\geq1$. Split the factors with indices
at least $j_0$ into blocks of $h$ consecutive factors and at most $h-1$
remaining factors. Each complete block has norm at most $1/2$.
At most $j_0+h-2$ factors lie outside these blocks. Their norms are
uniformly bounded by Corollary~\ref{cor:ric} and continuity of $F_c-KH_o$
on $\Theta$. Multiplying the bounds gives, for some $C_\Phi\geq1$,
\[
\|\Phi_{k,j}(\theta)\|_2 \leq C_\Phi\rho_\Phi^{k-j},
\quad \rho_\Phi=2^{-1/h},\quad k\geq j.
\]
By \eqref{eqn:time-invariant-predictor} and the bound on powers of
$F_c-KH_o$ in the proof of Lemma~\ref{lemma:convergence-riccati},
$\|\hat\xi_1(\theta)\|_2$ is bounded uniformly in $\theta$ by a constant
times the sum used in the proof of Lemma~\ref{lem:conv-quadrati-form}
at $k=1$, which is finite almost surely. The bound on $\Phi_{k,1}$
then makes $\sum_{k=1}^N\|\alpha_k\|_2^2$ bounded uniformly in
$\theta$ and $N$. Weighted Cauchy--Schwarz applied to the sum defining
$\beta_k$ gives a factor $\sum_{j=1}^{k-1}\rho_\Phi^{k-j-1}\leq
(1-\rho_\Phi)^{-1}$. Summing the remaining weights over $k$ gives
a second factor $(1-\rho_\Phi)^{-1}$, so
\[
\sum_{k=1}^N\|\beta_k\|_2^2
\leq \frac{C_\Phi^2\|H_o\|_2^2}{(1-\rho_\Phi)^2}
\sum_{j=1}^{N-1}\|\Delta K_j(\theta)\epsilon_j(\theta)\|_2^2.
\]
Since $\|\epsilon_j(\theta)\|_2^2\leq j C_\epsilon$, the last sum is
bounded by $C_\epsilon\sum_{j=1}^\infty j \sup_{\theta\in\Theta}\|\Delta K_j(\theta)\|_2^2$,
finite by Corollary~\ref{cor:ric}. Dividing the bounds on the two sums
of squares by $N$ proves \eqref{eqn:d-squared-limit}.
}

We have now proven that both \textcolor{black}{$\Delta'_N(\theta)$ and $\Delta''_N(\theta)$} converge to zero as $N$ goes to 
infinity uniformly in $\theta\in\Theta$, and hence by Lemma~\ref{lem:conv-quadrati-form} we have proven this 
lemma.
\end{IEEEproof}

\section*{Use of Generative AI}
During the revision of this work, OpenAI ChatGPT was used to help address
the concerns raised by the Associate Editor and reviewers in later rounds of
reviews. The authors reviewed, improved, and edited the content as needed
and take full responsibility for the content of the published article.

\bibliographystyle{plain}        
\bibliography{refs}    
%
\begin{IEEEbiography}[{\includegraphics[width=1in,height=1.25in,clip,keepaspectratio]{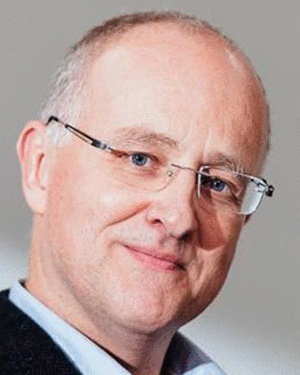}}]{Anders Hansson} received the Master of Science degree in electrical engineering, the Degree of Licentiate of Engineering in automatic control, and the Ph.D. degree in automatic control from Lund University, Lund, Sweden, in 1989, 1991, and 1995, respectively. 
From 1995 to 1997, he was a Postdoctoral Student, and from 1997 to 1998, a Research Associate with Information Systems Lab, Department of Electrical Engineering, Stanford University, Stanford, CA, USA. In 1998, he was appointed  Assistant Professor and in 2000,  Associate Professor with S3-Automatic Control, Royal Institute of Technology, Stockholm, Sweden. In 2001, he was appointed Associate Professor with the Division of Automatic Control, Linköping University, Linköping, Sweden. Since 2006, he has been a Full Professor with the same department. He
was a Visiting Professor at the Department of Electrical and Computer Engineering, UCLA, Los Angeles, USA during 2011-2012. 
His research interests include optimal control, stochastic control, linear systems, signal processing, applications of control, and telecommunications. He is currently a Member of the EUCA General Assembly, of the Technical Committee on Systems with Uncertainty of the IEEE Control Systems Society and the Technical Committee on Robust
Control of IFAC.
\end{IEEEbiography}
\begin{IEEEbiography}[{\includegraphics[width=1in,height=1.25in,clip,keepaspectratio]{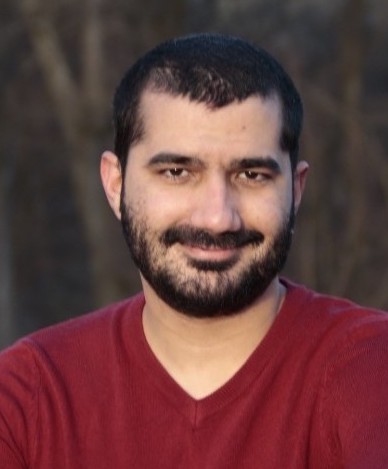}}]{João Victor Galvão da Mata} received his B.S. degree in electrical engineering from the Federal University of Bahia (UFBA), Brazil, in 2021, his M.S. degree from the University of Strasbourg, France, in 2022, and his Ph.D. degree in applied mathematics from the Technical University of Denmark (DTU), Denmark, in 2026. He is currently a Postdoctoral Researcher at Lund University, Sweden. His research interests include optimization, numerical methods, linear systems, automatic control, system identification, dynamical networks, and data-driven modeling. He was awarded the ``Grand Est Franco-German University Prize'' (Prix Universitaire Franco-Allemand Grand Est) in 2023 for the Best Master's Thesis in the category of Science and Technology.
\end{IEEEbiography}
\begin{IEEEbiography}[{\includegraphics[width=1in,height=1.25in,clip,keepaspectratio]{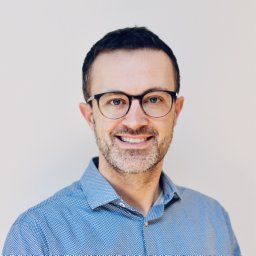}}]{Martin S.~Andersen} earned his MSc degree in Electrical Engineering from Aalborg University, Denmark, in 2006, and his PhD in Electrical Engineering from the University of California, Los Angeles (UCLA), USA, in 2011. Following his PhD, he served as a Postdoctoral Researcher at the Division of Automatic Control, Linköping University, Sweden, and the Technical University of Denmark (DTU). Currently, he is an Associate Professor in the Section for Scientific Computing, Department of Applied Mathematics and Computer Science at DTU. His research interests encompass optimization, numerical methods, signal and image processing, and systems and control.
\end{IEEEbiography}

\end{document}